\documentclass[10pt,twocolumn,twoside]{IEEEtran}
\usepackage[utf8]{inputenc}
\usepackage{listings}
\usepackage{amsmath,amssymb,amsthm}
\usepackage{graphicx,verbatim}
\usepackage[noadjust]{cite}
\usepackage{color}
\usepackage{url}
\usepackage{enumitem}
\usepackage{tcolorbox}
\usepackage{algorithm,algorithmic,float}
\usepackage[outdir=./figure]{epstopdf}
\usepackage{setspace}
\DeclareGraphicsExtensions{.pdf,.jpeg,.png}
\usepackage[caption=false, font=footnotesize]{subfig}
\usepackage{mathtools}

\makeatletter
\let\saveqed\qed
\renewcommand\qed{%
	\ifmmode\displaymath@qed
	\else\saveqed
	\fi}
\makeatother

\allowdisplaybreaks

\makeatletter
\let\NAT@parse\undefined
\makeatother
\usepackage[colorlinks=true, linkcolor=blue, citecolor=blue]{hyperref}

\tcbuselibrary{skins}

\newtheorem{theorem}{Theorem}[section]
\newtheorem{prop}[theorem]{Proposition}
\newtheorem{lemma}[theorem]{Lemma}
\newtheorem{corollary}[theorem]{Corollary}
\newtheorem{remark}[theorem]{Remark}

\makeatletter
\g@addto@macro\normalsize{%
\setlength\abovedisplayskip{4pt}
\setlength\belowdisplayskip{4pt}
}
\makeatother

\input{macros.tex}
\makeatletter
\newcommand\footnoteref[1]{\protected@xdef\@thefnmark{\ref{#1}}\@footnotemark}
\makeatother
\begin{document}
%
\title{Control Under Action-Dependent Markov Packet Drops: An
  Event-Triggered Approach}
\author{Shourya Bose \qquad Pavankumar Tallapragada %
  \thanks{This work was partially supported by Robert Bosch Centre for
    Cyber-Physical Systems, Indian Institute of Science, Bengaluru. A
    part of this work was presented in the Fifth Indian Control
    Conference as~\cite{SB-PT:2019-icc}. }%
  \thanks{S. Bose is with Department of EEE, BITS Pilani - KK Birla
    Goa Campus {\tt\small(f20140692@goa.bits-pilani.ac.in)} and P.
    Tallapragada is with the Department of Electrical Engineering,
    Indian Institute of Science {\tt\small (pavant@iisc.ac.in) }. }%
}
\maketitle
\raggedbottom
\begin{abstract}
  In this paper, we consider the problem of second moment
  stabilization of a scalar linear plant with process noise. We assume
  that the sensor must communicate with the controller over an
  unreliable channel, whose state evolves according to a Markov chain,
  with the transition matrix on a timestep depending on whether there
  is a transmission or not on that timestep. Under such a setting, we
  propose an event-triggered transmission policy which meets the
  objective of exponential convergence of the second moment of the
  plant state to an ultimate bound. Furthermore, we provide upper
  bounds on the transmission fraction of the proposed policy. The
  guarantees on performance and transmission fraction are verified
  using simulations.
\end{abstract}


\section{Introduction}

In networked control systems (NCS) feedback occurs over a
communication channel that may introduce a number of effects such as
sampling, packet drops and time delays. The resulting limitations on
the communication resources necessitates the design of parsimonious,
system aware communication and control. In this paper,
we study the problem of controlling a scalar linear system, using an
event-triggered approach, over an unreliable action-dependent Markov
channel.

\subsubsection*{Literature Review}

The last two decades have seen extensive research on NCS or control
over networks~\cite{JPH-etal:2007, PP-etal:2017-survey,
  ASM-AVS:2009-book}. Stochastic NCS over lossy or unreliable
communication channels is also an extensively studied
area~\cite{ASM-AVS:2009-book, LS-etal:2007, SY-TB:2013-book}. A common
assumption in this literature is that the packet drops are independent
and identically distributed (i.i.d.) on every timestep. However, some
works consider packet drop probabilities which evolve in time as a
Markov process. For example,~\cite{KY-MF-LX:2011, JW-etal:2018}
consider the problem of estimation and~\cite{LY-etal:2018} considers
the problem of stabilization of nonlinear systems, each under Markov
packet drops. References~\cite{KY-LX:2011, KO-HI:2014} explore the
problem of finite data rate control under Markovian packet losses. On
the other hand, \cite{PM-LC-MF:2013} considers the problem of control
over a channel that supports a finite data rate, with the data rate
evolving according to Markov chain. Finally, \cite{LX-LX-NX:2018} is
concerned with mean-square stabilization of a linear system under
Markov losses and Gaussian transmission noise,
while~\cite{LS-VG-GC:2019} considers mean-square stabilization over an
AWGN channel with fading subject to a Markov chain.

In the literature on communication systems, Markov models for channels
have a long history, starting with the work of Gilbert~\cite{ENG-1960}
and Elliott~\cite{EOE-1963}. The
paper~\cite{PS-etal:2008-markovwireless} is a relatively recent survey
on Markov modeling of fading channels.  Channels whose properties
depend on past actions have also been explored, including as models
for other applications. The reference~\cite{VJ-RJL-NCM:2019} is a
recent survey on models and research work on systems whose operation
depends on a ``utilization dependent component'' such as queueing in
action dependent servers, iterative learning algorithms and systems
with energy harvesting components, among other problems.

In this paper, the actions we seek to design are the transmission
times from a sensor to a controller. In particular, we take an
event-triggered approach, which in the last decade has been a very
popular method for parsimonious transmissions in
NCS~\cite{WPMHH-KHJ-PT:2012, ML:2010, DT-SH:2017-book}. However, the
volume of work on event-triggered control in a stochastic setting is
still not as considerable as in the deterministic setting. Some papers
that consider random packet drops are~\cite{BD-VG-DEQ-MJ:2017,
  MR-KHJ:2009, RB-FA:2012, MHM-DT-AM-SH:2014, VD-MH:2017},
while~\cite{DEQ-VG-WM-SY:14, RPA-DM-DVD:15} study stochastic stability
with event-triggered control. This paper builds upon our previous work
on event-triggered control under Bernoulli packet
drops~\cite{PT-MF-JC:2018-tac} and event-triggered control over Markov
packet drops~\cite{SB-PT:2019-icc}.

\subsubsection*{Contributions}

In this paper, we study the problem of second-moment stabilization of
a scalar linear plant with process noise over an unreliable
channel/network. In particular, the channel state determines the
packet drop probability and evolves according to a finite state space
Markov process that also depends on the past transmission actions.

Our first contribution is modeling an NCS with an action-dependent
Markov channel and design of a transmission policy over such a
channel. To the best of our knowledge, such channels have not been
considered before in the context of NCS.  Our second contribution is a
two-step design of event-triggered transmission policy. This design
approach is in a spirit similar to our earlier
work~\cite{PT-MF-JC:2018-tac, SB-PT:2019-icc}. We provide a necessary
condition on the plant dynamics and the channel parameters for our
transmission policy to work. This necessary condition is similar to
the conditions often found in the data rate limited
control~\cite{MF-PM:2014} and NCS in general. Our third contribution
is analysis of the proposed event-triggered transmission policy and
guarantee of second moment stability with an exponential convergence
to a desired ultimate bound. The fourth contribution is an upper bound
on the transmission fraction (the fraction of timesteps, in a time
duration, on which a transmission occurs) resulting from the
event-triggered policy. We provide upper bounds on the asymptotic
transmission fraction as well as for the `transient' transmission
fraction.

\subsubsection*{Notation}

We let $\real$, $\integers$, $\intpos$, and $\intnonneg$ denote the
sets of real numbers, integers, natural numbers and non-negative
integers, respectively. We use the standard font for scalar quantities
while boldface for vectors and matrices. The notations $\onevec$,
$\canbasis{i}$, and $\id$ denote the vector with all 1s, the vector
whose $\tth{i}$ takes the value $1$ and 0 everywhere else, and the
identity matrix, respectively, of appropriate dimensions. We use
$\specrad{\mb{A}}$ to denote the spectral radius of a real square
matrix $\mb{A}$. We denote the space of probability vectors
(i.e. vectors with non-negative entries that sum to $1$) of $n$
dimensions as $\svecspace$. The notation $\prob[.]$ denotes the
probability of an event. We denote a generic transmission policy using
$\tpolicy$, and $\expect{\tpolicy}{\;}$ represents expectation of a
random variable under a given transmission policy $\tpolicy$. We
denote the cardinality of a finite set $\Mcal{S}$ as $|\Mcal{S}|$. For
integers $a$ and $b$, we let $\intrangecc{a}{b}$, $\intrangeoo{a}{b}$,
and $\intrangeoc{a}{b}$ represent the finite sets
$[a,b]\cap\integers$, $(a,b)\cap\integers$, and $(a,b]\cap\integers$,
respectively. For random variables $X$, $Y$ and $Z$, the \emph{tower
  property of conditional expectation} is
\begin{equation*}
  \expect{}{\ \expect{}{X\,\vert\, Y, Z} \ \vert\, Y \ } =
  \expect{}{X\,\vert\, Y} .
\end{equation*}

\section{System Description}
\label{sec:SystemDescription}

In this section, we describe the model of the plant, channel,
controller and the control objective.

\subsection{Plant and Controller Model} 
\label{sec:plant-controller}

Consider a scalar linear plant with process noise
\begin{equation}
  x_{k+1}= a x_k + u_k + v_k, \quad  x_k,\,u_k,\,v_k \in \real, \;
  \forall k \in \intnonneg.
  \label{eq:plant_evolution}
\end{equation}
The parameter $a$ is the inherent plant gain, which we assume is
unstable, i.e. $|a|>1$. The variables $x_k$, $u_k$ and $v_k$ are the
plant state, the control input and the process noise, respectively at
timestep $k \in \intnonneg$. We assume that $v_k$ is independent and
identically distributed (i.i.d.) across timesteps $k$ and independent
of all the other system variables. Its distribution has zero mean and
finite variance, i.e.
$\expect{}{v_k} = 0, \, \expect{}{v_k^2} \rdef M>0$.

We assume that, at each timestep, a sensor perfectly measures the
plant state and can decide on whether to transmit a packet with the
plant state to the controller. We denote the sensor's transmission
decision on timestep $k$ by $\tk$ and we let
\begin{equation*}
  \tk \ldef \begin{cases}
    1,& \text{if sensor transmits at } k\\
    0,& \text{if sensor does not transmit at } k .
  \end{cases}
\end{equation*}
The sensor determines $\tk$ at each timestep $k$ according to an
event-triggered \emph{transmission policy} on the basis of plant state
and all the information available to it on timestep $k$. 
Even if the sensor transmits a packet at timestep $k$ ($\tk = 1$), the
packet may be dropped by the communication channel according to a
packet drop model which we describe in Section~\ref{sec:channel}. We
let $\rk$ be the reception indicator, which takes values as follows
\begin{equation*}
  \rk\ldef\begin{cases}
    1,& \text{if $\tk=1$ and packet received}\\
    0,& \text{if $\tk=1$ and packet dropped}\\
    0,& \text{if $\tk=0$} .
  \end{cases}
\end{equation*}

The controller maintains a \emph{controller state}, $\xkp$, which it
uses to generate the input $u_k \ldef L \xkp$, where $L$ is a constant
such that $\abar \ldef (a+L) \in (-1, 1)$. The controller state $\xkp$
itself evolves as
\begin{equation}
  \label{eq:controller_state_def}
  \xkp = \begin{cases}
    x_k,&\text{if } \rk=1\\
    \hat{x}_{k}, & \text{if }\rk=0 ,
\end{cases}
\end{equation}
where $\hat{x}_k \ldef \abar\hat{x}^+_{k-1}$ is the \emph{estimate}
of the plant state given past data. Corresponding to the controller
state and plant state estimate, we define the \emph{estimation error}
$z_k$ and \emph{controller state error} $z_k^+$ as follows.
\begin{equation}
  \label{eq:system_controller_errors_def}
  z_k\ldef x_k - \hat{x}_k, \quad z_k^+ \ldef x_k-\xkp .
\end{equation}
The two quantities differ only on successful reception times. It is
possible to write the plant state evolution equation in terms of these
errors as follows.
\begin{subequations} \label{eq:composite_evol}
	\begin{align}
    x_{k+1} &= a x_k+L\xkp+v_k = \abar x_k - L z_k^+ +
              v_k \label{eq:composite_evol_plant} 
    \\
    \hx_{k+1} &= \abar\hx^+_k . \label{eq:composite_evol_estimate}
	\end{align}
\end{subequations}
Equations~\eqref{eq:controller_state_def}-\eqref{eq:composite_evol}
compositely describe the evolution of the plant state, controller
state and the estimate of plant state.

\subsection{Channel Model} \label{sec:channel}

We let the communication channel be an action-dependent \emph{finite
  state space Markov channel} (FSSMC). We denote the \emph{channel
  state} at timestep $k$ by $\cstate \in \{ 1,\cdots,n \}$, with $n$ a
finite positive integer. We assume that the probability distribution
of $\gamma_{k+1}$ depends on $\cstate$ and $\tk$, the transmission
decision on timestep $k$. Thus, the evolution of the channel is an
action-dependent Markov process. We let $\pz_{ij}$ and $\po_{ij}$
denote the probabilities of the channel state transitioning from $j$
to $i$ given $\tk$ is equal to $0$ and $1$, respectively. Thus,
\begin{align*}
  \pz_{ij} &\ldef \prob \left[ \gamma_{k+1} = i \, \vert \,
          \cstate=j,\tk=0 \right]\\
  \po_{ij} &\ldef \prob \left[ \gamma_{k+1} = i \, \vert \,
          \cstate=j,\tk=1 \right] .
\end{align*}
We let $\pzero$ and $\pone$ be column-stochastic matrices, whose
$\tth{(i,j)}$ elements are $\pz_{ij}$ and $\po_{ij}$, respectively.

We model the unreliability of the channel through a packet drop
probability $e_i$ for each element $i$ of the channel state
space. Thus, if on timestep $k$ the channel state $\cstate = i$ and if
the sensor transmits a packet then the channel drops the it with
probability $e_i\in[0,1]$ and it communicates the packet successfully
to the controller with probability $(1 - e_i)$, i.e.,
\begin{equation*}
  \rk \ldef\begin{cases}
    1,& \text{w.p. $(1 - e_{\cstate})$ if $\tk=1$}\\
    0,& \text{w.p. $e_{\cstate}$ if $\tk=1$}\\
    0,& \text{if $\tk=0$} ,
  \end{cases}
\end{equation*}
where ``w.p.'' stands for ``with probability''. Thus, the packet drops
on each timestep is Bernoulli, though \emph{not} i.i.d.. We collect
the packet drop probabilities across all possible channel states in
the vector $\mb{e} \ldef [e_1,\ e_2,\ \cdots,\ e_n]^T \in
[0,1]^n$. Correspondingly, we define the transmission success
probability vector $\mb{d}$ as $\mb{d} \ldef \onevec-\mb{e}$.

\subsection{Sensor's Information Pattern}

Next, we describe the information available to the sensor to make the
transmissions decisions $\tk$. Apart from the plant state $x_k$ that
the sensor can measure perfectly on each timestep $k$, we assume that
if a successful reception occurs on timestep $k$, then the controller
acknowledges it by relaying the reception indicator variable $\rk$ and
the channel state $\cstate$ over an error-free feedback channel, which
the sensor may use the information only on subsequent timesteps.

To describe all the information available to the sensor on timestep
$k$ more formally, we first introduce the variables $\Rk$ and $\Rk^+$
to track the \emph{latest reception time before} and \emph{latest
  reception time until} timestep $k$, respectively. Thus,
\begin{equation*}
  \Rk \ldef \max\limits_i \left\{ i<k  :  r_i=1 \right\}, \;
  \Rk^+ \ldef \max\limits_i \left\{ i\leq k : r_i=1 \right\} .
\end{equation*}
The variable $\Rk$ is useful for the sensor's decision making while
$\Rk^+$ is helpful in the analysis. Further, we let $\Sj$ for
$j\in \intnonneg$ be the $\tth{j}$ successful random reception time,
that is,
\begin{equation*}
  S_0 = 0, \; \Sjnext \ldef \min \left\{ k>\Sj : \rk=1 \right\}, \, \forall
  \, j \in \intpos ,
\end{equation*}
where without loss of generality, we have assumed that the zeroth
successful reception occurs on time $0$.

From the controller feedback, the sensor knows $R_k$ and
$\gamma_{\Rk}$ before deciding $t_k$, from which the sensor can
utilize the channel evolution model to obtain the probability
distribution of the channel state $\pk \in \svecspace$ given $\Rk$,
$\gamma_{\Rk}$ and all the transmission decisions from $\Rk$ to $k-1$,
that is,
  \begin{equation*}
    \pk(i) \ldef \prob \left[ \cstate = i \ | \ \Rk, \
      \gamma_{\Rk}, \ \{ t_w \}_{\Rk}^{k-1} \right] ,
  \end{equation*}
  where $\pk(i)$ is the $\tth{i}$ element of the vector
  $\pk$. Thus, we can obtain $\pk$ recursively as
  \begin{equation}
  \mb{p}_{k+1} =
  \begin{cases}
    \pone \canbasis{\cstate}, & \text{if } \tk = 1 \text{ and } \rk =
    1\\
    \pzero \pk,& \text{if } \tk=0 \text{ and } \rk=0\\
    \pone \pk,& \text{if } \tk=1 \text{ and } \rk =0 .
  \end{cases}
  \label{eq:channel-state-prob}
\end{equation}
We also let
\begin{equation*}
  \pk^+ \ldef%
  \begin{cases}
    \delta_{\gamma_k}, \quad &\text{if } \rk = 1 \\
    \pk, \quad &\text{if } \rk = 0 .
  \end{cases}
\end{equation*}

We represent by $\Ik$ the information available to the sensor about
the controller's knowledge of plant state before transmission while we
use $\Ik^+$ to denote the information available to the sensor after
channel state feedback (if any). Thus, $\Ik^+=\Ik$ when $\rk=0$, and
$\Ik^+$ contains $\rk$ and $\cstate$ over $\Ik$ when $\rk=1$. Noting
the same, we define $\Ik$ and $\Ik^+$ as
\begin{subequations}
  \begin{align}
    \Ik &\ldef \{
          k,x_k,z_k,\Rk,x_{\Rk},\pk, t_{k-1},r_{k-1}\gamma_{k-1}
          \}, \label{eq:Ik_def} \\
    \Ik^+ &\ldef \{
            k,x_k,z_k^+,\Rk^+,x_{\Rk^+},\pk^+,\tk,\rk\cstate \}.
            \label{eq:Ikplus_def}     
  \end{align}
\end{subequations}
Note that the channel state feedback by the controller is represented
as $r_{k-1} \gamma_{k-1}$ and $\rk \cstate$ in $\Ik$ and $\Ik^+$,
respectively. If $\rk=1$ then $\rk\cstate=\cstate$, and if $\rk=0$
then $\rk\cstate=0$ and thus no channel state feedback is available.
Note that $\{\Ik\}_{k \in \intnonneg}$ and
$\{\Ik^+\}_{k \in \intnonneg}$ are action-dependent Markov
processes. In particular, the probability distribution of $I_k$
conditioned on $\{ I_s, t_{s} \}_{s=0}^{k-1}$ can be shown to be the
same as the one conditioned on $\{I_{k-1},t_{k-1}\}$. Similarly,
$\{I_k^+ \}$ is ``sufficient information'' to determine the
distribution of $I_{k+1}^+$ given all the past information


\subsection{Control Objective}

The control objective is exponential second-moment stabilization of
the plant state to an ultimate bound. Given the plant and the
controller models in Section~\ref{sec:plant-controller}, the only
decision making left to be designed is the sensor's transmission
policy $\tpolicy$, which determines $\tk$ for each timestep $k$. In
particular, we seek to design a feedback transmission policy using the
available information $\Ik$ on timestep $k$. The \emph{offline control
  objective} that we seek to guarantee is
\begin{equation}
  \label{eq:original-control-objective}
  \expect{\tpolicy}{x_k^2 \, \vert \, I_0^+} \leq \max \{
  c^{2k}x_0^2,B \}, \; \forall k \in \intnonneg ,
\end{equation}
which is to have the second moment of the plant state decay
exponentially at least at a rate of $c^2$ until it settles to the
ultimate bound $B$. We assume that the convergence rate parameter
$c^2 \in (\abar^2, 1)$. Note
that~\eqref{eq:original-control-objective} prescribes the restriction
on the plant state evolution in an offline fashion, in terms of only
the initial information. However, a recursive formulation of the
control objective is more conducive to designing a feedback
transmission policy.

To design a feedback transmission policy, we need to define an online
version of the control objective. First, we define the
\emph{performance function} $h_k$ for every timestep $k$ as follows
\begin{equation*}
  h_k \ldef x_k^2 - \max \{ c^{2(k-\Rk)}x_{\Rk}^2,B \}.
\end{equation*}
Then, the \emph{online objective} is to ensure
\begin{equation}
  \expect{\tpolicy}{h_k^2 \, \vert \, I_{\Rk}^+} \leq 0, \quad \forall
  k \in \intnonneg .
  \label{eq:online-obj}
\end{equation}
We borrow from Lemma~III.1 from~\cite{PT-MF-JC:2018-tac}, which
demonstrates that any transmission policy that satisfies the online
objective also satisfies the offline objective. 

\begin{lemma}[Sufficiency of the online
  objective~\cite{PT-MF-JC:2018-tac}]
  If a transmission policy $\tpolicy$ satisfies the online
  objective~\eqref{eq:online-obj} then it also satisfies the offline
  objective~\eqref{eq:original-control-objective}. \qed
  \label{lem:online}
\end{lemma}

Note that in the control
objective~\eqref{eq:original-control-objective}, the sources of
randomness that determine the expectation are the transmission policy
$\tpolicy$, the random channel behavior and the process noise. The
transmission policy and the random channel behavior determine the
successful reception times while the process noise affects the
evolution of the performance function during the inter-reception
times. As the online objective~\eqref{eq:online-obj} is essentially a
condition on the evolution of the performance function during the
inter-reception times, Lemma~\ref{lem:online} continues to hold in the
setting of this paper.

\section{Two-Step Design of Transmission Policy}
\label{sec:twostep}

Designing a transmission policy so that the described system meets the
control objective~\eqref{eq:original-control-objective} or even the
stricter online objective~\eqref{eq:online-obj} poses many
challenges. The main challenge stems from the random packet drops,
which makes the necessity of a transmission on timestep~$k$ dependent
on future transmission decisions. Furthermore, the future evolution of
the channel state depends on all the past and current transmission
decisions. Thus, the transmission decisions $\tk$ cannot be made in a
myopic manner and instead must be made by evaluating their impact on
the channel and the control objective over a sufficiently long time
frame.  To tackle this problem, we adopt a two-step design
procedure. This general design principle is the same as
in~\cite{PT-MF-JC:2018-tac}, wherein the reader can find a more
detailed discussion about this procedure as well as its merits. We now
describe the two steps of the design procedure.

In the first step, for each timestep $k$, we consider a family of
\emph{nominal policies} with \emph{look-ahead parameter
  $D \in \intpos$}. A nominal policy with parameter $D$ involves a
`hold-off' period of $D$ timesteps from $k$ to $k+D-1$ during which
$\tk=0$, and then there is perpetual transmission, that is $\tk=1$ for
all timesteps after $k+D-1$. Thus, letting $\tkd$ be the nominal
policy with parameter $D$, we can formally express it as
\begin{equation}
  \tkd: t_i=
  \begin{cases}
    0,& \text{if } i \in \{ k,k+1,\cdots,k+D-1 \}\\
    1,& \text{for } i \geq k+D .
  \end{cases}
  \label{eq:nominal-policy}
\end{equation}

In the second step of the design procedure, we construct the
event-triggered policy, $\etpol$, using the nominal policies as
building blocks. Given~\eqref{eq:nominal-policy}, one can reason that
if the nominal policy with parameter $D \in \intpos$ satisfies the
online objective from the current timestep $k$, then a transmission on
the current timestep is not necessary to meet the online
objective. Further, if the online objective cannot be met from
timestep~$k$ using the nominal policy $\tkd$ then it may be necessary
to transmit on timestep $k$. This forms the basis for the construction
of the event-triggered policy, which we detail next.

First, we need a method to check if the nominal policy $\tkd$
satisfies the online objective from timestep~$k$. For this, we define
the \emph{look-ahead function}, $\gdk$, as the expected value of the
performance function $h_k$ at the next successful reception timestep
$k=\Sjnext$ under the nominal policy, that is,
\begin{equation}
  \label{eq:lookahead_crit_initial_def}
  \gdk \ldef \expect{\tkd}{h_{\Sjnext}\,\vert\,\Ik, \Sj =
    \Rk} .
\end{equation}
We can evaluate $\gdk$ as a total expectation, over all possible
values of $\Sjnext$, as
\begin{align}
  &\gdk = \notag \\
  &\sum_{w=D}^{\infty} \expect{\tkd}{h_{\Sjnext}\,\vert\,\Ik, \Sj =
    \Rk, \Sjnext = k+w} \Om{D}(w,\pk) ,
\label{eq:gdk-expand}
\end{align}
where $\Om{D}(w,\mb{p})$ is the probability of the event that the
first successful reception after timestep $k$ is at timestep $k+w$
under the nominal policy $\tkd$ and given $\mb{p}$, the probability
distribution of the channel state at time $k$, conditioned on the
information at time $\Rk$. Formally,
\begin{align}
  \Omega_{D}(w,\mb{p}) \ldef \prob [ \Sjnext = k+w \,|\, \tpolicy
  = \tkd, \pk = \mb{p}, \Sj = \Rk ] .
    \label{eq:omega_prob_def}
\end{align}
The closed form of $\Om{D}(w,\mb{p})$ is given as follows.
\begin{align}
  &\Omega_{D}(w,\mb{p}) = \dt(\ponee)^{(w-D)}
    \pzero^{(D)} \mb{p} ,
    \label{eq:omega-prob-closed-form}
\end{align}
where $\mb{E}$ is the diagonal matrix with elements of $\mb{e}$ on its
main diagonal. The explanation of~\eqref{eq:omega-prob-closed-form} is
as follows - the probability vector $\mb{p}$, when left-multiplied by
$\pzero^{(D)}$ provides the probability vector of the channel state
immediately after the hold-off period, which is of $D$ timesteps. The
said vector when left-multiplied by $(\ponee)^{(w-D)}$ provides the
probabilities of, subsequent to the hold-off period, making a
transmission attempt $(w-D)$ times successively but failing to achieve
reception on every attempt. Finally, left-multiplication by $\dt$
gives the probability of finally having a successful reception on the
$\tth{(k+w)}$ timestep. Thus~\eqref{eq:omega-prob-closed-form} is the
closed form of $\Omega_{D}(w,\mb{p})$ defined
in~\eqref{eq:omega_prob_def}.

\subsection{The Event-Triggered Policy}

We now describe the event-triggered policy. The main idea behind the
proposed event-triggered policy is the following. A negative sign of
the look-ahead function~$\gdk$ indicates that it is not ``necessary''
to transmit on timestep $k$ as there exists a transmission sequence
(given by the nominal policy) that meets the objective at least on the
next random reception timestep. However, if the sign of $\gdk$ is
non-negative, it means that the sensor cannot afford to hold off
transmission for $D$ timesteps from the current timestep $k$, and
still ensure that the online objective is not violated on some future
timestep. In the proposed event-triggered transmission policy, the
sensor evaluates $\gdk$ at every timestep $k$, and when it turns
nonnegative the sensor keeps transmitting on every timestep until a
successful reception occurs, and then the sensor again waits for
$\gdk$ to turn non-negative. The event-triggered transmission policy
may be described formally as follows.
\begin{equation}
  \label{eq:et_policy}
  \etpol : \tk = 
\begin{cases}
  0, & \text{if } k \in \{ R_k + 1, \cdots, \tau_k - 1 \}\\
  1, & \text{if } k \in \{ \tau_k, \cdots, Z_k \} ,
\end{cases}
\end{equation}
where $\tau_k$ is the first timestep after $R_k$ when $\gdk \geq 0$
and $Z_k$ is the first timestep, after $\Rk$, on which there is a
successful reception. Thus, formally,
\begin{align*}
  \tau_k & \ldef \min\{ m > \Rk \, : \, \G{m}{D} \geq 0\} ,
  \\
  Z_k & \ldef \min\{ m > \Rk \, : \, R_m^+ = m \} .
\end{align*}
Note that the event-triggered policy is described recursively in
terms of $R_k$, the latest reception time before $k$, and the
look-ahead function~$\gdk$. As a result, the policy
in~\eqref{eq:et_policy} is valid over the entire infinite time
horizon. In the analysis of the policy~\eqref{eq:et_policy} in the
sequel, it is useful to refer to the $\tth{j}$ reception time, denoted
by $\Sj$. Similarly, we let
\begin{equation*}
  T_j \ldef \min\{ m > \Sj \, : \, \G{m}{D} \geq 0\} .
\end{equation*}
So, if $\Sj = R_k$ then $T_j = \tau_k$ and $\Sjnext = Z_k$.

\section{Implementation and Performance Guarantees}
\label{sec:PerformanceGuarantee}

In this section, we describe the implementation details of the
proposed event-triggered policy, and analyze the system under this
policy through several intermediate results. At the end of the
section, we provide sufficient conditions on the ultimate bound $B$
and the look-ahead parameter $D$ such that the system meets the online
objective (and hence the offline objective) under the event-triggered
policy.

\subsection{Closed Form Expression of the Look Ahead Criterion}

For implementation of the event-triggered policy~\eqref{eq:et_policy},
we need an easy method to compute the look-ahead function $\gdk$. In
particular, we provide here a closed form expression of the look-ahead
function. We begin by expanding the expectation term
in~\eqref{eq:gdk-expand} as follows~\cite{PT-MF-JC:2016-allerton}
\begin{align}
  &\expect{}{h_{\Sjnext} \, \vert \, I_k, \Sj = \Rk, \Sjnext=k+w} =
    \notag\\
  &\abar^{2w}x_k^2 +2\abar^w(a^w-\abar^w)x_kz_k +
    (a^{2w}-2a^w\abar^w + \abar^{2w})z_k^2\notag\\
  & +\bM (a^{2w}-1) -\max \{ c^{2w} c^{2(k-\Rk)} x^2_{\Rk},B \} .
		\label{eq:expect-h}
\end{align}
From~\eqref{eq:gdk-expand} and~\eqref{eq:expect-h}, it is evident that
convergence of $\gdk$ requires the convergence of infinite series of
the form
\begin{align}
  g_D(b,\mb{p}) & \ldef \sum_{w=D}^\infty b^{w} \Om{D}(w,\mb{p})
                  \notag
  \\
                & = b^D \sum_{w=D}^\infty b^{(w-D)} \dt (\ponee)^{(w-D)}
                  \pzero^{(D)}\mb{p} ,
\label{eq:gD-def}
\end{align}
with $\mb{p} \in \svecspace$, and $D \in \intpos$ and for values of $b$
equal to $\abar^2$, $c^2$, $a^2$, $\abar a$ and $1$, which satisfy
\begin{equation}
  0 < \abar^2 < c^2 < 1 < a^2, \quad | \abar a | < a^2 .
  \label{eq:param-ordering}
\end{equation}
Each of the terms~$g_D(b,\mb{p})$ is an infinite matrix geometric
series. The criteria for convergence and the closed form of
$g_D(b,\mb{p})$ for these values of $b$ would allow us to determine
the same for $\gdk$. Thus, our first aim is to examine the condition
for convergence of~$g_D(b,\mb{p})$. The following lemma establishes
the necessary and sufficient condition for a matrix geometric series
to converge. Although this is a well known result, we could not find a
concise proof in the literature and hence we present a proof ourselves
in Appendix~\ref{sec:aux-proofs}.

\begin{lemma}[Convergence of Matrix Geometric Series]
	Consider a square matrix $\mb{K}$. The matrix geometric series
  $\sum_{w=0}^\infty \mb{K}^w$ converges if and only if
  $\specrad{\mb{K}}<1$. Further, if $\specrad{\mb{K}}<1$ then the
  series converges to $(\id - \mb{K})^{-1}$. \qed
	\label{lem:matrix-geom-series}
\end{lemma}

We can obtain a closed form expression of $g_D(b,\mb{p})$ defined
in~\eqref{eq:gD-def} by first expressing it as
\begin{equation*}
  g_D(b,\mb{p}) = b^ D \dt \left[\sum_{w=0}^\infty
    (b\mb{\ponee})^w \right] \pzero^{(D)}\mb{p} .
\end{equation*}
and then applying Lemma~\ref{lem:matrix-geom-series}. In particular if
$\specrad{b\mb{\ponee}} < 1$ then we obtain
\begin{equation} \label{eq:gD-closed-form}
  g_D(b,\mb{p}) = b^D \dt \infsum{b} \mb{P}_0^{(D)} \mb{p} .
\end{equation}
In the following result, we apply Lemma~\ref{lem:matrix-geom-series}
to provide a necessary and sufficient condition for $\gdk$ to be
well-defined. Its proof appears in Appendix~\ref{sec:aux-proofs}.

\begin{lemma}[Necessary and sufficient condition for the existence of
  $\gdk$]
	$\gdk$ converges for all values of the probability distribution
  vector $\pk$ \emph{if and only if} $a^2\specrad{\ponee} < 1$. \qed
	\label{lem:gdk-convergence}
\end{lemma} 

We now proceed to give a closed form expression of the look-ahead
function $\gdk$ in the following lemma. Its proof is presented in
Appendix~\ref{sec:aux-proofs}.

\begin{lemma}[Closed form of the look-ahead function]
  \label{lem:jdsjclosed}
  Suppose that $a^2\specrad{\ponee} < 1$. The following is a
  closed-form expression of the look-ahead function $\gdk$.
  \begin{align*}
    \gdk = %
    & g_D(\abar^2,\pk)x_k^2 + 2\left(g_D(a\bar{a},\pk)-g_D
      (\bar{a}^2,\pk)\right)x_kz_k+\notag\\
    & \left(g_D(a^2,\pk) + g_D(\bar{a}^2,\pk) -
      2g_D(a\bar{a},\pk) \right) z_k^2 + \notag\\
    & \bM \left( g_D(a^2,\pk)-g_D
      (1,\pk)\right) - \left( B f_D(1,\pk)
      \right. + \notag\\
    & \left. N_k\left[g_D(c^2,\pk)
      -f_D(c^2,\pk) \right] \right)
  \end{align*}
  where $\bM \ldef M(a^2-1)^{-1}$, $N_k \ldef c^{2(k-\Rk)} x^2_{\Rk}$,
  the closed form of the function~$g_D(b,\mb{p})$ is given
  in~\eqref{eq:gD-closed-form}, while $f_D(b,\mb{p})$ is given by
  \begin{align*}
    f_D(b,\mb{p}) &\ldef b^{\mu} \dt (\mb{P}_1\mb{E})^{(\mu - D)}
                    \infsum{b} \mb{P}_0^{(D)}\mb{p} .
  \end{align*}
  Finally, $\mu$ is defined as follows
  \begin{align}
    \mu &\ldef \max \left\{ D, \left\lceil \frac{\log
            (x^2_{\Rk}/B)}{\log(1/c^2)} \right\rceil -(k-\Rk)
            \right\} .
            \label{eq:qkD}
  \end{align}
  \qed
\end{lemma}

Note that the closed form of $\gdk$ is a third-degree polynomial of
the plant state $x_k$, error $z_k$, and individual elements of $\pk$,
and is amenable for online computation. Furthermore, note that the
look-ahead function $\gdk$ possesses a mathematical structure
comprising of a linear operator with unit dimensional rowspace acting
on the stochastic vector $\pk$.

\subsection{Necessary Condition on the Ultimate Bound $B$}

We now seek a necessary condition on the ultimate bound $B$ for there
to exist a transmission policy that satisfies the online objective. To
this end, we introduce the \emph{open loop performance function},
$H(w,y)$, which we define as the expectation of the performance
function $h_{\Sjnext}$ conditioned upon $I_{\Sj}^+$ and the event that
$\Sjnext=\Sj+w$ and $x_{\Sj}^2 = y$, that is,
\begin{equation}
  H(w, y) \ldef
  \expect{}{h_{\Sjnext}\,\vert\,I_{\Sj}^+, x_{\Sj}^2 = y, \Sjnext=\Sj+w}.	
  \label{eq:OLPF_definition}
\end{equation}
Note that $H(w, x_{\Sj}^2)$ is very similar to~\eqref{eq:expect-h}
except that $H$ is conditioned upon $I_{\Sj}^+$ and defined for the
special case of $k = \Sj$. Thus, the closed form of $H(w, x_{\Sj}^2)$
may be obtained from~\eqref{eq:expect-h} by replacing $k$ with $\Sj$,
$x_k$ with $x_{\Sj}$ and $z_k$ with $z_{\Sj}^+ = 0$ and $R_k$ with
$R_{\Sj}^+ = \Sj$. Hence we have
\begin{equation}
  H(w,x^2_{\Sj}) = \abar^{2w}x^2_{\Sj} + \bar{M}(a^{2w}-1)
  -\max\{c^{2w}x^2_{\Sj},B\}.
  \label{eq:H-expr}
\end{equation}
Note that $H(w,x^2_{\Sj}) < 0$ indicates that given the information
$I_{\Sj}^+$, the online objective is expected to be satisfied on
timestep $\Sj + w$. Conversely, a positive sign implies that the
online objective is expected to be violated on timestep $\Sj +
w$. Using this observation, we demonstrate in the following
proposition that for $B$ less than a critical $B_0$, there exists
\emph{no} transmission policy that can satisfy the online
objective. We provide its proof in Appendix~\ref{sec:aux-proofs}.

\begin{prop}[Necessary condition on the ultimate bound for meeting the
  online objective]
	If $B<B_0\ldef\frac{\bM\log(a^2)}{\log(c^2/\abar^2)}$ then no
  transmission policy satisfies the online objective. \qed
	\label{prop:open_loop_performance_function}
\end{prop}

Proposition~\ref{prop:open_loop_performance_function} demonstrates
that $B>B_0$ is a \emph{necessary} condition on $B$ for a transmission
policy to satisfy the online objective. In the following subsection,
we further analyse the open-loop performance function $H(w,y)$ to find
a \emph{sufficient} criterion to check whether a given $B$ and $D$
ensure that the online objective is met under the event-triggered
policy.

\subsection{The Performance-Evaluation Function, $\jdsj$}

For the purpose of analysing system performance betweeen any two
successive reception times $\Sj$ and $\Sjnext$, we define the
\emph{performance-evaluation function}, $\jdsj$.  It is defined
similarly as $\gdk$ in~\eqref{eq:lookahead_crit_initial_def}, though
only for $k=\Sj$ (successful reception times) and conditioned upon the
information set $I^+_{\Sj}$ instead of $I_{\Sj}$. In particular, we
let
\begin{align}
  \jdsj \ldef%
  & \ \mathbb{E}_{\tpolicy_{\Sj+1}^{D-1}} \left[ h_{\Sjnext}\, \vert \,
    I_{\Sj}^+\right] = \sum_{w=D}^\infty H(w,x^2_{\Sj})
    \ti{\Omega}_D(w,\gamma_{\Sj}).
\label{eq:jdsj_infinite_sum_form}
\end{align}
Here, $\ti{\Omega}_D(w,\gamma)$ denotes the probability of getting a
successful reception $w$ timesteps after $S_j$ starting with channel
state $\gamma$ on $\Sj$ under the nominal policy
$\Mcal{T}^{D-1}_{\Sj+1}$. The purpose of the function
$\ti{\Omega}_D(w,\gamma)$ is analogous to that of $\Omega_D(w,\mb{p})$
in $\gdk$, and is formally defined as
\begin{align}
  \hspace{-0.5em} \ti{\Omega}_{D}(w,\gamma) \ldef \prob [ \Sjnext =
  \Sj+w \,|\, \tpolicy = \Mcal{T}^{D-1}_{\Sj+1},\gamma_{\Sj}=\gamma] .
  \label{eq:omega_2_prob_def}
\end{align}
Note that there are two differences between the closed forms of
$\Om{D}(w,\mb{p})$ and $\ti{\Omega}_D(w,\gamma)$. First, we use
channel state $\gamma$ instead probability vector $\mb{p}$ in the
definition of $\Omega_D(w,\gamma)$, since channel state distribution
$\mb{p}_{\Sj}=\canbasis{\gamma_{\Sj}}$ can be inferred from
$I^+_{\Sj}$ using~\eqref{eq:channel-state-prob}. Second, the
expectation in~\eqref{eq:jdsj_infinite_sum_form} is conditioned upon
nominal policy $\Mcal{T}^{D-1}_{\Sj+1}$ as opposed to the policy
$\Mcal{T}^D_{\Sj}$ in the definition of $\Mcal{G}_{\Sj}^D$
in~\eqref{eq:gdk-expand}. This is because $\gdk$ is defined
\emph{pre-transmission} for the purpose of deciding $\tk$, while
$\jdsj$ is defined \emph{post-transmission} on timestep $\Sj$ for the
purpose of convergence analysis. Therefore $\gdk$ is calculated for a
timestep with an underlying nominal policy in which $t_i=0$ for
$i\in\{k,k+1,\cdots,k+D-1\}$, while the definition of $\jdsj$ is
already conditioned upon the fact that $t_{S_j}=1$ for all $j$. The
closed form of $\ti{\Omega}_D(w,\gamma)$ can be obtained similarly to
the closed form of $\Omega_D(w,\mb{p})$, and is given as
\begin{equation}
\ti{\Omega}_D(w,\gamma)= \dt (\ponee)^{(w-D)}\pzero^{(D-1)}\pone\canbasis{\gamma}.
\label{eq:OmegaTildeClosedForm}
\end{equation}
For a well-chosen value of $B$, it can be shown that the open loop
performance function possesses the property of \emph{sign
  monotonicity}. This property is an important characteristic of
$H(w,y)$ and will prove useful in later results. 

\begin{theorem}[Sign behaviour of the open-loop performance function,
  Proposition IV.6,~\cite{PT-MF-JC:2018-tac}]
  \label{th:sign_monotonicity_H}
	There exists a $B^* \geq B_0$ with $B_0$ defined in
  Proposition~\ref{prop:open_loop_performance_function} such that if
  $B>B^*$, then $H(w,y)>0$ implies $H(s,y)>0$ for all $s\geq w$. \qed
\end{theorem}

The value of $B^*$ defined in Theorem~\ref{th:sign_monotonicity_H} can
be numerically computed using the procedure in
Appendix~\ref{sec:Bstar}, which is based on the proof of Lemma IV.13
in~\cite{PT-MF-JC:2018-tac}. We now provide a closed form expression
of the performance evaluation function $\jdsj$, similar to the closed
form of $\gdk$ in Lemma~\ref{lem:jdsjclosed}. The proof appears in
Appendix~\ref{sec:aux-proofs}.

\begin{lemma}[Closed form of performance-evaluation function]
		\label{lem:jdsjclosed2}
	 Suppose
  that $a^2\specrad{\ponee} < 1$. A closed form of the
  performance-evaluation function $\jdsj$ is given as
	\begin{align*}
    \J{\Sj}{D}\ldef %
    &\ti{g}_D(\bar{a}^2,\gsj) x_{\Sj}^2 + \bM\left[
      \ti{g}_D(a^2,\gsj) -\ti{g}_D(1,\gsj) \right]
      \notag\\
    &- \left[B\ti{f}_D(1,\gsj) + x^2_{\Sj} \left( \ti{g}_D(c^2,\gsj) -
      \ti{f}_D(c^2,\gsj) \right) \right] ,
	\end{align*}
	where
	\begin{align*}
    \ti{f}_D(b,\gamma) \ldef\, %
    & b^{\nu} \dt(\ponee)^{(\nu-D)}
      \infsum{b}
      \mb{P}_0^{(D-1)}\mb{P}_1\canbasis{\gamma},\\ 
    \ti{g}_D(b,\gamma) \ldef\, %
    & b^D \dt \infsum{b}
      \mb{P}_0^{(D-1)}\mb{P}_1\canbasis{\gamma}, 
	\end{align*}
  and finally, $\nu$ is defined as
	$$\nu \ldef \max \left\{ D,\left\lceil \frac{\log (x^2_{\Sj}/B)}{\log
    (1/c^2)} \right\rceil \right\}. \qed$$
\end{lemma}
The next result is concerned with the expected value of $\G{k+1}{D}$
after no transmission or after successful reception and the channel
state feedback on timestep $k$. Note that this result is valid for
\emph{any} transmission policy $\tpolicy$.

\begin{theorem}[Expected value of look-ahead function on next timestep]
	Let $\tpolicy$ be any transmission policy. Then, the following hold.
	\begin{enumerate}[label=(\alph*)]
  \item
    $\mathbb{E}_{\tpolicy} \left[ \G{k+1}{D} \, \vert \,I_k,
      \tk=0 \right] = \G{k}{D+1}$.
		\label{item:itma}
  \item
    $\mathbb{E}_{\tpolicy} \left[ \G{k+1}{D} \, \vert \, I_k,
      \rk=1,\cstate \right] = \J{k}{D+1}$. 
		\label{item:itmb}
	\end{enumerate}
	\label{th:NEXT}
\end{theorem}

\begin{proof}
  \textbf{\ref{item:itma}:} Note that
	\begin{align*}
    &\expect{\Mcal{T}}{\G{k+1}{D}\,\vert\,I_k,t_k=0}\\*
    &\;\; \stackrel{[r1]}{=} \expect{\Mcal{T}}{
      \expect{\Mcal{T}^D_{k+1}}{h_{\Sjnext} 
      \,\,\vert\, I_{k+1},S_{j}=R_{k+1}} \,\vert\,I_k,\tk=0},\\
    &\;\;
      \stackrel{[r2]}{=} \expect{\Mcal{T}^{D+1}_k}{\expect{\Mcal{T}^D_{k+1}}{h_{\Sjnext}
      \,\,\vert\, I_{k+1},S_{j}=R_{k}}\,\vert\,I_k,\tk=0},\\ 
    &\;\;
      \stackrel{[r3]}{=} \expect{\Mcal{T}^{D+1}_k}{h_{\Sjnext}\,\vert\,I_{k},\tk=0,\Sj=R_k} =
      \G{k}{D+1} ,
	\end{align*}
  where [r1] follows from~\eqref{eq:lookahead_crit_initial_def}, while
  in [r2] we can replace the policy $\tpolicy$ with $\Mcal{T}^{D+1}_k$
  because the event $\tk = 0$ is consistent with the policy
  $\Mcal{T}_k^{D+1}$ on time step $k$ and once $\tk = 0$ is fixed the
  expected value of $\G{k+1}{D}$ is independent of the transmission
  policy used on subsequent timesteps. In [r2], we also use the fact
  that if $\tk=0$ then $R_{k+1} = R_k$. Finally, [r3] uses the fact
  that $\{I_k,\tk\}$ is \emph{sufficient information}
  and then the tower property.
	
	\textbf{\ref{item:itmb}:}
  For proving this part, we observe that $I_k$ and the
  additional information that $\rk = 1$ and $\gamma_k$ implies the knowledge of
  $I_k^+$. Considering this fact and proceeding with a
  similar methodology as the proof of claim~\ref{item:itma}, we observe that
	\begin{align*}
    &\expect{\Mcal{T}}{\G{k+1}{D}\,\vert\,I_k, \rk=1, \gamma_k }\\
    &\;\; =\expect{\Mcal{T}}{\expect{\Mcal{T}^D_{k+1}}{h_{\Sjnext}
      \,\,\vert\, I_{k+1},S_j=R_{k+1}}\,\vert\, I_k^+, \rk = 1 },\\ 
    &\;\;
      =\expect{\Mcal{T}^{D}_{k+1}}{\expect{\Mcal{T}^D_{k+1}}{h_{\Sjnext}
      \,\,\vert\, I_{k+1},S_j=R_{k+1}}\,\vert\,I_k^+, S_j = k},\\ 
    &\;\;
      =\expect{\Mcal{T}^{(D+1)-1}_{k+1}}{h_{\Sjnext}\,\vert\,I_k^+,
      S_j = k}
      = \J{S_j}{D+1}.\hspace{5.5em}\qedhere
	\end{align*}
\end{proof}

We use Theorem~\ref{th:sign_monotonicity_H},
Lemma~\ref{lem:jdsjclosed2}, and Theorem~\ref{th:NEXT} to provide a
\emph{sufficient} condition on ultimate bound $B$, and the look-ahead
parameter $D$ such that the event-triggered policy meets the online
objective. First, in Proposition~\ref{prop:PEF_upper_bound}, we obtain
an upper bound on $\J{\Sj}{D}$ that is uniform in $x_{\Sj}$ and
depends only on the channel state $\gamma_{\Sj}$. We present its proof
in Appendix~\ref{sec:aux-proofs}. Subsequently, we give a sufficient
condition to ensure the upper bound, and hence $\J{\Sj}{\theta}$, to
be negative.

\begin{prop}[Upper bound on $\J{\Sj}{\theta}$]
	\label{prop:PEF_upper_bound}
  For the look-ahead parameter $\theta\in\intpos$, the performance
  evaluation function $\J{\Sj}{\theta}$ is uniformly (in $x_{S_j}$)
  upper bounded as $\J{\Sj}{\theta} \leq \Mcal{R}_j(\theta)$, where
	\begin{align*}
    \hspace{3em}
    \Mcal{R}_j(\theta) \ldef%
    & \left[ \ti{g}_\theta(\abar^2,\gsj) - \ti{g}_\theta(c^2,\gsj)
      \right]\frac{B}{c^{2\theta}} \\
    &+ \bM \left[ \ti{g}_\theta(a^2,\gsj) - \ti{g}_\theta(1,\gsj)
      \right]. \qed
      \hspace{3em}
	\end{align*}
\end{prop}

Having established a uniform upper bound on the performance-evaluation
function in the last result, we now provide a \emph{sufficient}
condition on the system parameters
such that $\J{\Sj}{\theta} <0$ for all values of $x_{S_j}$ and $\gsj$.

\begin{theorem}[Sufficient condition for negativity of
  performance-evaluation function]
  Suppose
  \begin{equation*}
    B \geq B_0 = \frac{\bM\log(a^2)}{\log(c^2/\abar^2)} .
  \end{equation*}
  Consider the vector valued function
  $\mb{Q}( \theta ) : \real \rightarrow \real^n$ given by
	\begin{equation*}
    \mb{Q}(\theta) \ldef \left[
      \Mcal{Z}_\theta(\abar^2)-\Mcal{Z}_\theta(c^2) \right]
    \frac{B}{c^{2\theta}} + \bM \left[ \Mcal{Z}_\theta(a^2) -
      \Mcal{Z}_\theta (1) \right]
	\end{equation*}
	wherein
  $\Mcal{Z}_\theta(b) \ldef b^\theta \dt (\mb{I}-b\ponee)^{-1}$. If
  $\mb{Q}(D)<\mb{0}$ (elementwise), for some $D\in\intpos$, then
  $\J{\Sj}{\theta} < 0$ for all $x_{S_j} \in \real$ and for all
  $\theta\in\until{D}$.
	\label{th:pef_satisfied}
\end{theorem}

\begin{proof}
	We start the proof by noting that $\Mcal{R}_j(\theta)$ from
  Proposition~\ref{prop:PEF_upper_bound} can be written as
	\begin{equation} \label{eq:QR}
    \Mcal{R}_j(\theta) = \mb{Q}(\theta) \pzero^{(\theta-1)} \pone
    \canbasis{\gsj} .
	\end{equation}
  From the elementwise non-negativity of
  $\pzero^{(\theta-1)} \pone \canbasis{\gsj}$ for all
  $\theta\in\intpos$ and $\gsj\in\until{n}$, we conclude that a
  \emph{sufficient} condition to ensure $\Mcal{R}_j(D) < 0$ for a
  given $D$ and all $j\in\intnonneg$ is to ensure that
  $\mb{Q}(D)<\mb{0}$. We now show that every element of
  $\mb{Q}(\theta)$ is monotonically increasing in $\theta$, and thus,
  $\mb{Q}(D)<\mb{0}$ ensures $\mb{Q}(\theta)<\mb{0}$ for
  $\theta\in\until{D}$. The first and the second derivatives of
  $\mb{Q}(\theta)$ with respect to $\theta$ are
	\begin{align*}
    \frac{ \mathrm{d} \mb{Q}(\theta) }{ \mathrm{d} \theta } %
    &= \frac{ B }{ c^{2 \theta} } \log \left(
      \frac{\abar^2}{c^2}
      \right) \Mcal{Z}_{\theta}(\abar^2) + \bM \log(a^2)
      \Mcal{Z}_{\theta} (a^2)
    \\
    \frac{ \mathrm{d}^2 \mb{Q}(\theta) }{ \mathrm{d} \theta^2 } %
    &=
      \frac{B}{c^{2\theta}} \log^2\left(
      \frac{\abar^2}{c^2} \right)
      \Mcal{Z}_\theta(\abar^2) + \bM
      \log^2(a^2)\Mcal{Z}_\theta(a^2) .
	\end{align*}
  Note that each element of the second derivative is strictly
  positive. Thus, each element of $\mb{Q}(\theta)$ is strictly convex
  in $\theta$. Also, note that the first derivative of
  $\mb{Q}(\theta)$ at $ \theta=0$ is
  \begin{align*}
    \frac{ \mathrm{d} \mb{Q}(\theta) }{ \mathrm{d} \theta }
    &\stackrel{[r1]}{>} B\log \left( \frac{c^2}{\abar^2} \right)
      \left[ \Mcal{Z}_0 (a^2) - \Mcal{Z}_0 (\abar^2) \right] > 0,
  \end{align*}
  where [r1] follows from the fact that $B \geq B_0$. Since each
  element of $\mb{Q}(\theta)$ is strictly convex for
  $\theta \in \real$ and increasing at $\theta=0$, it follows that
  each element of $\mb{Q}(\theta)$ is monotonically increasing for
  $\theta \geq 0$. Thus, $\mb{Q}(D)<0$ implies $\mb{Q}(\theta)<0$, and
  thereby $\J{\Sj}{\theta}<0$ for all $\theta\in\until{D}$.
\end{proof}

We consolidate the results so far to provide a theoretical performance
guarantee that the event-triggered policy satisfies the online
objective~\eqref{eq:online-obj}.

\begin{theorem}[Performance guarantee of the event-triggered policy]
	\label{th:performance_final_guarantee}
  If $B > B^*$ (see Appendix~\ref{sec:Bstar}) and the lookahead
  parameter $D$ satisfies the condition $\mb{Q}(D)<\mb{0}$ then the
  event-triggered policy~\eqref{eq:et_policy} guarantees that the
  online objective~\eqref{eq:online-obj}, and therefore the original
  offline objective~\eqref{eq:original-control-objective}, are met.
\end{theorem}

\begin{proof}

  Given Lemma~\ref{lem:online}, it suffices to show that the online
  objective~\eqref{eq:online-obj} is met by the event-triggered
  policy. We center the proof around the following two claims.

  \emph{Claim (a):} For any $j \in \intnonneg$,
  $\expect{\etpol}{ h_{\Sjnext} \exsep I_{\Sj}^+ } \leq 0$ implies
  $\expect{\etpol}{ h_{k} \exsep I_{\Sj}^+ } \leq 0$ for all
  $k \in \intrangecc{\Sj}{\Sjnext}$.
	
	\emph{Claim (b):} For any $j \in \intnonneg$,
	$\expect{\etpol}{ h_{\Sjnext} \exsep I_{\Sj}^+ } < 0$.

  Together these two claims guarantee that the online objective is
  met, because
	\begin{align*}
    &\expect{\etpol}{ h_k \exsep I_0^+ }
    \\
    &= \expect{\etpol}{ \ldots \expect{\etpol}{ \expect{\etpol}{
      h_k \exsep I_{\Sj^+} } \exsep I_{S_{{j}-1}}^+ } \ldots \exsep
      I_0^+ } ,
	\end{align*}
  where $\{S_i\}$ are the random reception times and $\Sj = \Rk^+$.

	To prove Claim (a), we note that by the definition of open-loop
  performance function $H(w,y)$ in~\eqref{eq:OLPF_definition}, we have
	\begin{align*}
    \expect{\etpol}{h_k \, \vert \, I_{\Sj}^+ } =
    H(k-\Sj,x^2_{\Sj}),\; \forall k \in \intrangecc{\Sj}{\Sjnext} .
	\end{align*}
	If
  $\expect{\etpol}{h_{\Sjnext} \, \vert \, I^+_{\Sj}} =
  H(\Sjnext-\Sj,x^2_{\Sj}) < 0$, then the sign monotonicity property
  of the open-loop performance function
  (Theorem~\ref{th:sign_monotonicity_H}) implies
  $H(k-\Sj,x^2_{\Sj}) \leq 0$ for all
  $k \in \intrangecc{\Sj}{\Sjnext}$, which proves Claim (a).
	
	We now prove Claim (b). It can be seen from Theorem~\ref{th:NEXT}
  that for all $k \in \intrangeoo{\Sj}{T_j}$,
    \begin{align}
    \label{eq:claimbstep1}
      &\expect{\etpol}{ \G{k+1}{D} \exsep k \in \intrangeoo{\Sj}{T_j},
        \ I_{\Sj}^+ } \notag\\
      & \stackrel{[r1]}{=} \expect{\etpol}{ \expect{\etpol}{ \G{k+1}{D}
        \exsep I_k, t_k = 0} \exsep I_{\Sj}^+ } \notag\\
      & \stackrel{[r2]}{=} \expect{\etpol}{
        \G{k}{D+1} \exsep I_{\Sj}^+ } ,
    \end{align}
    where [r1] is obtained by using the tower property and the fact
    that $t_k = 0$ for $k \in \intrangeoo{\Sj}{T_j}$, while [r2] is
    obtained from Theorem~\ref{th:NEXT}. Furthermore,
    Theorem~\ref{th:NEXT}~(b) implies that
	\begin{align}
    \expect{\etpol}{ \G{\Sj+1}{D} \exsep I_{\Sj}^+ } %
    &=
      \expect{\etpol}{ \G{\Sj+1}{D} \exsep I_{\Sj}, r_{\Sj} =
      1,\gamma_{\Sj} } \notag \\
    & = \J{\Sj}{D+1}.
	\label{eq:claimbstep2}
	\end{align}
  Next, we condition the expected value of $h_{\Sjnext}$ over
  information from timestep $T_j$ as well as timestep $\Sj$ and using
  the tower property of conditional expectations, we obtain
	\begin{align}
    &\expect{\etpol}{h_{\Sjnext} \,\vert\, I^+_{\Sj}} \notag\\%
    &\stackrel{[r3]}{=} \expect{\etpol}{
      \expect{\Mcal{T}_{T_j}^0}{h_{\Sjnext} \, \vert
      \, I_{T_j}, S_j = R_{T_j} } \, \vert \, I^+_{\Sj} } \notag\\
    &= \expect{\etpol}{\G{T_j}{0} \, \vert \, I^+_{\Sj} }
	\label{eq:final_proof_step1}
	\end{align}
	where the inner expectation in [r3] is conditioned under the nominal
  policy $\Mcal{T}^0_{T_j}$ since for all timesteps
  $k \in \intrangecc{T_j}{\Sjnext}$, we have transmissions ($\tk=1$).
  We consider two cases: $T_j\leq \Sj+D$ and $T_j>\Sj+D$. In the first
  case, since $\tk=0$ for $k \in \intrangeoo{S_j}{T_j}$, we
  use~\eqref{eq:claimbstep1} and~\eqref{eq:claimbstep2} to
  write~\eqref{eq:final_proof_step1} as
	\begin{equation*}
    \expect{\etpol}{\G{T_j}{0} \, \vert \, I^+_{\Sj} } =
    \expect{\etpol}{\G{\Sj+1}{T_j-\Sj-1} \, \vert \, I^+_{\Sj}} =
    \J{S_j}{T_j-\Sj},
  \end{equation*}
	where Theorem~\ref{th:pef_satisfied} ensures that if
  $T_j - \Sj \leq D$ then $\J{\Sj}{T_j-\Sj} < 0$. We now consider the
  second case in which $T_j > \Sj + D$. Since we have $\tk=0$ for
  $k \in \intrangeoo{S_j}{T_j}$, we use~\eqref{eq:claimbstep1} to
  write~\eqref{eq:final_proof_step1} as
	\begin{equation*}
	\expect{\etpol}{\G{T_j}{0} \, \vert \, I^+_{\Sj}} = \expect{\etpol}{\G{T_j-D}{D} \, \vert \, I^+_{\Sj}} < 0,
	\end{equation*}
	since $\G{k}{D}$ is negative, by definition, for
  $k \in \intrangeoo{S_j}{T_j}$. This proves Claim (b), and hence also
  the result.
\end{proof}

\section{Transmission Fraction}
\label{sec:TransmissionFraction}

This section analyzes the efficiency of the proposed event-triggered
transmission policy in terms of the fraction of times the sensor
transmits ($\tk=1$) over a given time horizon. First, we introduce the
\emph{transmission fraction} up to timestep $K$ as
\begin{equation*}
  \tf{K} \ldef \frac{\expect{\etpol}{\sum_{i=1}^K t_i \;
    \big\vert\; I_0^+}}{\expect{\etpol}{K \, \big\vert \, I_0^+}},
\end{equation*}
wherein the stopping timestep $K$ could itself be a random variable.
We call the limit of $\tf{K}$ when $K \rightarrow \infty$ as the
\emph{asymptotic transmission fraction}, denoted by $\tf{\infty}$.

We also consider another type of transmission fraction which we call
the \emph{transmission fraction up to state $\sx$}, and denote it with
$\tfs{\sx}$. It is defined as the transmission fraction up to the
first reception timestep such that the squared plant state is lesser
than $\sx$. That is,
\begin{equation*}
  \tfs{\sx} \ldef \frac{\expect{\etpol}{\sum_{i=1}^{S_{j}} t_i \;
      \big\vert\; I_0^+, \ \{x_{S_l}^2 \}_{l=0}^{j-1} \geq \sx, \
      x^2_{\Sj}<\sx}}{\expect{\etpol}{S_{j} \; \big\vert \; I_0^+, \ 
      \{ x^2_{S_l} \}_{l=0}^{j-1} \geq \sx, \ x^2_{\Sj}<\sx}}.
\end{equation*}

In Theorem~\ref{th:stateTF}, we provide an upper bound on $\tfs{\sx}$
which only involves plant and channel parameters, and $\sx$. From this
result, we derive an upper bound on the asymptotic transmission
fraction $\tf{\infty}$ as a corollary. Together, these results form a
figure-of-merit to determine channel utilization for different values
of plant and channel parameters, as well as the operational value $D$
of the look-ahead parameter.

\begin{theorem}[Upper bound on $\tfs{\sx}$]
	\label{th:stateTF}
	Suppose $\mb{Q}(D)<\mb{0}$ for a given value of $D$. The \tfsxText
  is upper bounded by
	\begin{equation*}
    \tfs{\sx} \leq \frac{\cone}{\czero+\cone},
	\end{equation*}
	where
  \begin{align*}
    &\czero \ldef \argmax\limits_{\Mcal{B}\in\intnonneg} \{
      \mb{Q}_{\sx}(D+\Mcal{B}) < \mb{0} \} \\
    &\mb{Q}_{\sx} (\theta) \ldef \\
    &\left[ \Mcal{Z}_\theta(\abar^2) - \Mcal{Z}_\theta(c^2)\right]
      \max \{ \sx,Bc^{-2\theta} \} + \bM \left[ \Mcal{Z}_\theta(a^2) -
      \Mcal{Z}_\theta(1) \right],
	\end{align*}
	with $\Mcal{Z}_\theta(b)$ as defined in
  Theorem~\ref{th:pef_satisfied}, while $\cone$ is given by
	$$
	\cone = \max\limits_{i \in \until{n}} \{ \dt
  (\ponee)(\mb{I}-\ponee)^{-2}\canbasis{i} \}.
	$$
\end{theorem}

\begin{proof}
	We find an upper bound on $\tfs{\sx}$ by first considering the time
  horizon between two successive reception times, and then extending
  the analysis to an arbitrary number of inter-reception cycles. For
  $j\in\intnonneg$, we let $\horizon_j$ be the time horizon
  $\intrangeoc{\Sj}{\Sjnext}$. Further, throughout this proof, we use
  the shorthand
  $\matrixterms{ \theta }{ \gsj } \ldef
  \pzero^{(\theta-1)}\pone\canbasis{\gsj}$ for notational convenience.
	
	Using the structure of the event-triggered policy, we split
  $\horizon_j$ into two parts as
  $\horzero_j \ldef \intrangeoo{\Sj}{T_j}$ and
  $\horone_j \ldef \intrangecc{T_j}{\Sjnext}$. Hence, for
  $k\in\horzero_j$, no transmission occurs ($\tk=0$) while for each
  $k\in\horone_j$, a transmission occurs ($\tk=1$). Now, consider the
  following two claims.

\emph{Claim (a):}
$\expect{\etpol}{|\horzero_j| \, \big\vert \, I_{\Sj}^+, \ x^2_{\Sj}
  >\Mcal{X}} \geq \czero$.
	
  \emph{Claim (b):}
  $\expect{\etpol}{ |\horone_j| \, \big\vert \, I_{\Sj}^+} \leq\cone$,
  for all $x_{\Sj} \in \real$.

  Supposing the two claims are true, consider the transmission
  fraction during the $\tth{j}$ horizon, $\horizon_j$, conditioned on
  $I_{\Sj}^+$
  \begin{align*}
    &\frac{\expect{\etpol}{|\horone_j| \, \big\vert \, I_{\Sj}^+}}{
      \expect{\etpol}{|\horzero_j| \, \big\vert \, I_{\Sj}^+} +
      \expect{\etpol}{|\horone_j| \, \big\vert \, I_{\Sj}^+}} \leq
      \frac{\cone}{\czero+\cone}
  \end{align*}
  since the transmission fraction is increasing in the term
  $\expect{\etpol}{|\horone_j| \, \big\vert \, I_{\Sj}^+}$, and
  decreasing in the term
  $\expect{\etpol}{|\horzero_j| \, \big\vert \, I_{\Sj}^+}$. Now, as
  this upper bound is independent of the state of the system as long
  as $ x^2_{\Sj} > \Mcal{X}$, we obtain the upper bound on
  $ \tfs{\sx}$, stated in the result. Thus all that remains now is to
  prove claims (a) and (b).
  
  To prove Claim~(a), we start by demonstrating that, for a given
  value of $\theta\in\intpos$ and under the assumption that
  $x^2_{\Sj}\geq\sx$,
  $\J{\Sj}{\theta} \leq
  \mb{Q}_{\sx}(\theta)\matrixterms{\theta}{\gsj}$. To this end, we
  consider two cases, $\sx \in \llow = [0,Bc^{-2\theta})$ and
  $\sx \in \lhi = [Bc^{-2\theta},\infty)$ respectively. If
  $\sx \in \llow$, then we have
  \begin{align}
    \notag
    \J{\Sj}{\theta} \leq \Mcal{R}_j(\theta) = \mb{Q}(\theta)
    \matrixterms{ \theta }{ \gsj } =
    \mb{Q}_{\sx}(\theta)\matrixterms{\theta}{\gsj} ,
  \end{align}
  where the inequality is from Proposition~\ref{prop:PEF_upper_bound},
  the first equality from~\eqref{eq:QR} and the second eqaulity from
  the fact that $\sx \in \llow$. Now, consider the case of
  $x^2_{\Sj} \geq \sx \in \lhi$. Recall from the proof of
  Proposition~\ref{prop:PEF_upper_bound} that
  \begin{align*}    
    \J{\Sj}{\theta} \leq%
    &\left[ \ti{g}_\theta(\abar^2,\gsj) - \ti{g}_\theta(c^2,\gsj)
      \right] x^2_{\Sj} + \bM [
      \ti{g}_\theta(a^2,\gsj) \\
    &- \ti{g}_\theta(1,\gsj) ]\\    
    \stackrel{[r1]}{=}& \bigg[ \left( \Mcal{Z}_\theta(\abar^2) -
                        \Mcal{Z}_\theta(c^2) 
                        \right) \max\{\sx,Bc^{-2\theta}\} + \\
    &\bM \left( \Mcal{Z}_\theta(a^2) - \Mcal{Z}_\theta(1)\right)
      \bigg] \matrixterms{\theta}{\gsj}
      \stackrel{[r2]}{=} \ \mb{Q}_{\sx}(\theta) \matrixterms{\theta}{\gsj} ,
  \end{align*}
  where [r1] is a result of \eqref{eq:QR} and the facts
  that $\Mcal{Z}_\theta(\abar^2)-\Mcal{Z}_\theta(c^2) < \mb{0}$, and
  $x^2_{\Sj} \geq \sx \geq Bc^{-2\theta}$, and [r2] uses
  the definition of $\mb{Q}_{\sx}(\theta)$. Thus, we have demonstrated
  that for any given $\sx \geq 0$, if $x^2_{\Sj} \geq \sx$ then
  $\J{\Sj}{\theta} \leq \mb{Q}_{\sx}(\theta)
  \matrixterms{\theta}{\gsj}$.

  Now, suppose $x^2_{\Sj} \geq \sx$ and
  $\mb{Q}_{\sx}( D + \Mcal{B} ) < \mb{0}$ for some
  $\Mcal{B} \in \intnonneg$, where $D$ is the operational value of the
  look-ahead parameter. Then, through a recursive application of
  Theorem~\ref{th:NEXT} $\Mcal{B}$ times, we get
  \begin{align}
    \expect{\etpol}{\G{\Sj + \Mcal{B}}{D} \vert I_{\Sj}^+} 
    =
    \expect{\etpol}{\J{\Sj}{D+ \Mcal{B}} \vert I_{\Sj}^+ } \leq
    \mb{Q}_{\sx}( D + \Mcal{B} ) < \mb{0}
    . \label{eq:intermediate2_tfstate}
  \end{align}
  Hence, from the design of the event-triggered
  policy~\eqref{eq:et_policy}, it follows that $T_j>\Sj+\Mcal{B}$, or
  in other words, no transmission takes place at least $\Mcal{B}$
  timesteps from $\Sj$, \emph{in expectation}. Thus,
  \begin{align*}
    \expect{\etpol}{ |\horzero_j| \big\vert I_0^+, x^2_{\Sj} \geq \sx} &\geq
    \czero.
  \end{align*}
	We now consider Claim (b). Note that $\tk = 1$ for all
  $k \in \horone_j$, and by the structure of the event-triggered
  policy, $\expect{\etpol}{ |\horone_j| }$ is simply the expected
  number of timesteps for reception under a string of continuous
  transmission attempts, starting from timestep $T_j$ and channel
  state $\gamma_{T_j}$. To capture the same, we define the constant
  $\cone_i$ for $i\in\until{n}$ as
	\begin{align*}
	\cone_i & \ldef \, \mathbb{E} [ w-T_j \, \vert \, w \geq T_j :
            r_{w}=1,t_w=1, r_{w-1}=0,\\
	& t_{w-1}=1,\cdots,r_{T_j}=0,t_{T_j}=1 \, \vert \, \gamma_{T_j}=i]\\
	& = \dt\left[ \sum_{s=0}^\infty s(\ponee)^s \right] \canbasis{i} =
   \dt(\ponee)(\mb{I}-\ponee)^{-2}\canbasis{i}.
	\end{align*}
	We bound $|\horone_j|$ by simply choosing the highest value of
  $\cone_i$ among $i\in\until{n}$, thereby showing that $\cone$ is
  indeed an upper bound on $|\horone_j|$, and proving Claim (b) and
  hence also the result.
\end{proof}

Note that the term $\czero$ in the upper bound on $\tfs{\sx}$ is basically the $\Mcal{B}$-maximizer of $\mb{Q}_{\sx}(D+\Mcal{B})$ under the constraint that $ \mb{Q}_{\sx}(D+\Mcal{B}) < \mb{0}$. This fact illuminates the trade-off between control performance and transmission fraction, which we highlight in the following remark.

\begin{remark}[Tradeoff between control performance and transmission fraction]
	Suppose for a given value of $\sx$ and some $\psi \in \intpos$, we
  have $\mb{Q}_{\sx}(\psi)<\mb{0}$ but
  $\mb{Q}_{\sx}(\psi+1)\canbasis{i} \geq 0$ for at least one
  $i\in\until{n}$. Then if the operational value of the look-ahead
  parameter is $D$, we note that $D+\Mcal{B} = \psi$. The system
  designer can either choose a high value of $D$ (conservative
  control) but this results in a lower value of $\Mcal{B}$, and thus a
  larger upper bound on $\tfs{\sx}$. Conversely, a lower value of $D$
  (aggressive control) leads to a higher value of $\Mcal{B}$, and thus
  a smaller upper bound on $\tfs{\sx}$. \bulletend
	\label{rem:perf_tradeoff}
\end{remark}

We show in the following result that an upper bound on the asymptotic
transmission fraction, $\tf{\infty}$ can be obtained by setting
$\sx=Bc^{-2D}$ in the upper bound of $\tfs{\sx}$ provided in
Theorem~\ref{th:stateTF}. We present its proof in
Appendix~\ref{sec:aux-proofs}.

\begin{corollary}[Upper bound on asymptotic transmission fraction]
	The asymptotic transmission fraction $\tf{\infty}$ is upper bounded by
	\begin{equation*}
	\tf{\infty} \leq \frac{\cone}{\czero_\infty + \cone},
	\end{equation*}
	where
	\begin{align*}
    &\czero_\infty \ldef \argmax\limits_{\Mcal{B}\in\intnonneg} \{
      \mb{Q} (D+\Mcal{B}) < \mb{0} \}
	\end{align*}
	and $\cone$ as defined in Theorem~\ref{th:stateTF}. \qed
	\label{cor:infTF}
\end{corollary}

%

\section{Simulations}
\label{sec:Simulations}

In this section, we present simulation results to validate the
event-triggered policy. We choose a scalar plant with $a = 1.10$, with
desired convergence rate $c = 0.98$ and $\abar = 0.95c$. We choose the
following $\pzero$ and $\pone$ matrices.
\begin{align*}
  \pzero= 
  \begin{bmatrix}
    0.5&0.4&0.4&0.3\\
    0.3&0.3&0.2&0.3\\
    0.2&0.2&0.2&0.3\\
    0.0&0.1&0.2&0.1
  \end{bmatrix}
                 \!,\;
                 \pone=
                 \begin{bmatrix}
                   0.1&0.0&0.1&0.2\\
                   0.1&0.1&0.2&0.2\\
                   0.1&0.3&0.3&0.3\\
                   0.7&0.6&0.4&0.3
                 \end{bmatrix} .
\end{align*}
The matrices $\pzero$ and $\pone$ are chosen to have the
characteristic that the probability of going to a higher numbered
state is \emph{lower} in $\pzero$, and \emph{higher} in $\pone$, as
compared to staying in the same or going to a lower numbered state.
Correspondingly, we choose the following values of packet drop
probability for channel states 1 through 4
\begin{equation*}
  \mb{e}=
  \begin{bmatrix}
    0.1&0.2&0.3&0.4
  \end{bmatrix}^{T} .
\end{equation*}

We choose the process noise $v_k$ randomly with independent and
identical normal distribution having mean $0$ and variance $M =
1$. For the chosen parameters, the procedure in
Appendix~\ref{sec:Bstar} gives $B^*=8.8483$. We choose $B = 18 > B^*$
and pick the initial plant state as $x_0 = 15.5B$.
\subsection*{Simulation Results}

We simulated the above system, using MATLAB, for various values of
$D$. We obtained $10000$ empirical iterations for each value of
$D$. Figure~\ref{f:1a} demonstrates the effect of the choice of the
look-ahead parameter $D$ on the closed-loop plant state evolution. For
a higher value of $D$, the policy $\etpol$ is more conservative in the
sense that the value of squared plant state is in general lesser and
farther away from the desired envelope than for a higher value of $D$
(refer to Remark~\ref{rem:perf_tradeoff}). Note that in
Figure~\ref{f:1a}, the squared plant state in general has a lower
value with $D=2$ than with $D=1$.

Figure~\ref{f:1b} demonstrates the empirical value of $\tf{k}$ for two
values of $D$ over a horizon of $5,000$ timesteps. Note that very
large $k$, the transmission fraction reaches an asymptotic value.
Higher values of $D$ lead to a greater steady state value of the
transmission fraction. Moreover, we also calculate and present the
theoretical upper bound on transmission fraction for both $D=1$ and
$D=2$, using Corollary~\ref{cor:infTF}. Figure~\ref{f:1c} depicts
empirically calculated $\tfs{\sx}$, \tfsxText, for $D=1$, $D=2$, and
$D=3$.  To generate the empirical value of $\tfs{\sx}$ the $\sx$-axis
was bucketed with buckets of exponentially increasing size, since
finding a continuous empirical relation between $\sx$ and $\tfs{\sx}$
would be impossible within a finite number of simulated trajectories.
We again see that higher values of $D$ result in higher transmission
fractions.

Figure~\ref{f:1d} is a plot of the theoretical upper bounds on both
$\tf{\infty}$ and $\tfs{\sx}$, the asymptotic transmission fraction
and the transmission fraction up to state $\sx$, respectively, for
various values of the look-ahead parameter $D$. From this figure, it
can be visually verified that the upper bound on $\tf{\infty}$ given
by Corollary~\ref{cor:infTF} is the same as the upper bound on
$\Mcal{F}_{Bc^{-2D}}$ given by Theorem~\ref{th:stateTF}, for different
values of $D$.

\begin{figure}[t!]
	\centering \subfloat[\label{f:1a}]{%
		\includegraphics[width=0.47\linewidth]{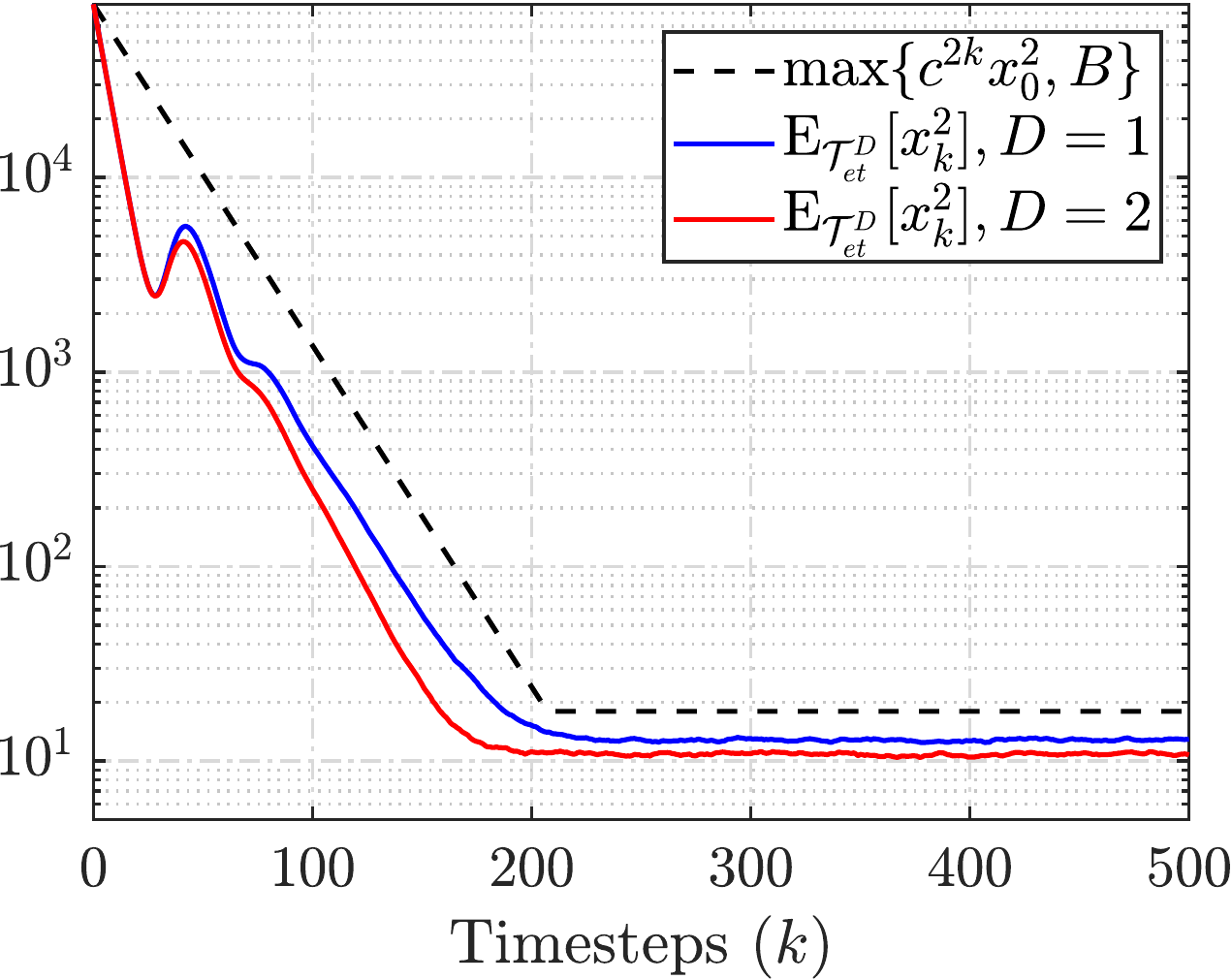}}
  \hspace{0.04\linewidth}%
  \subfloat[\label{f:1b}]{%
		\includegraphics[width=0.47\linewidth]{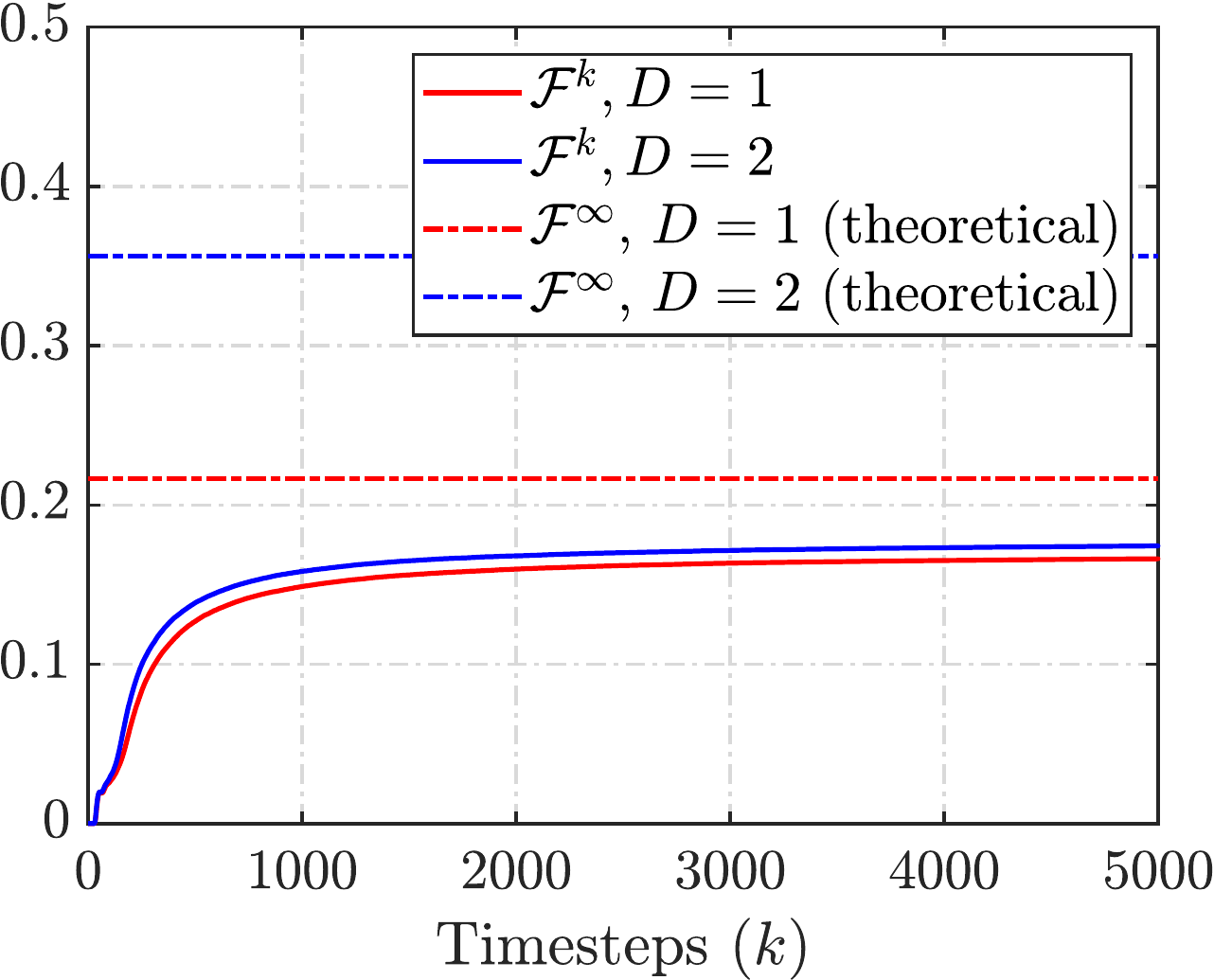}}
	\\
	\subfloat[\label{f:1c}]{%
		\includegraphics[width=0.47\linewidth]{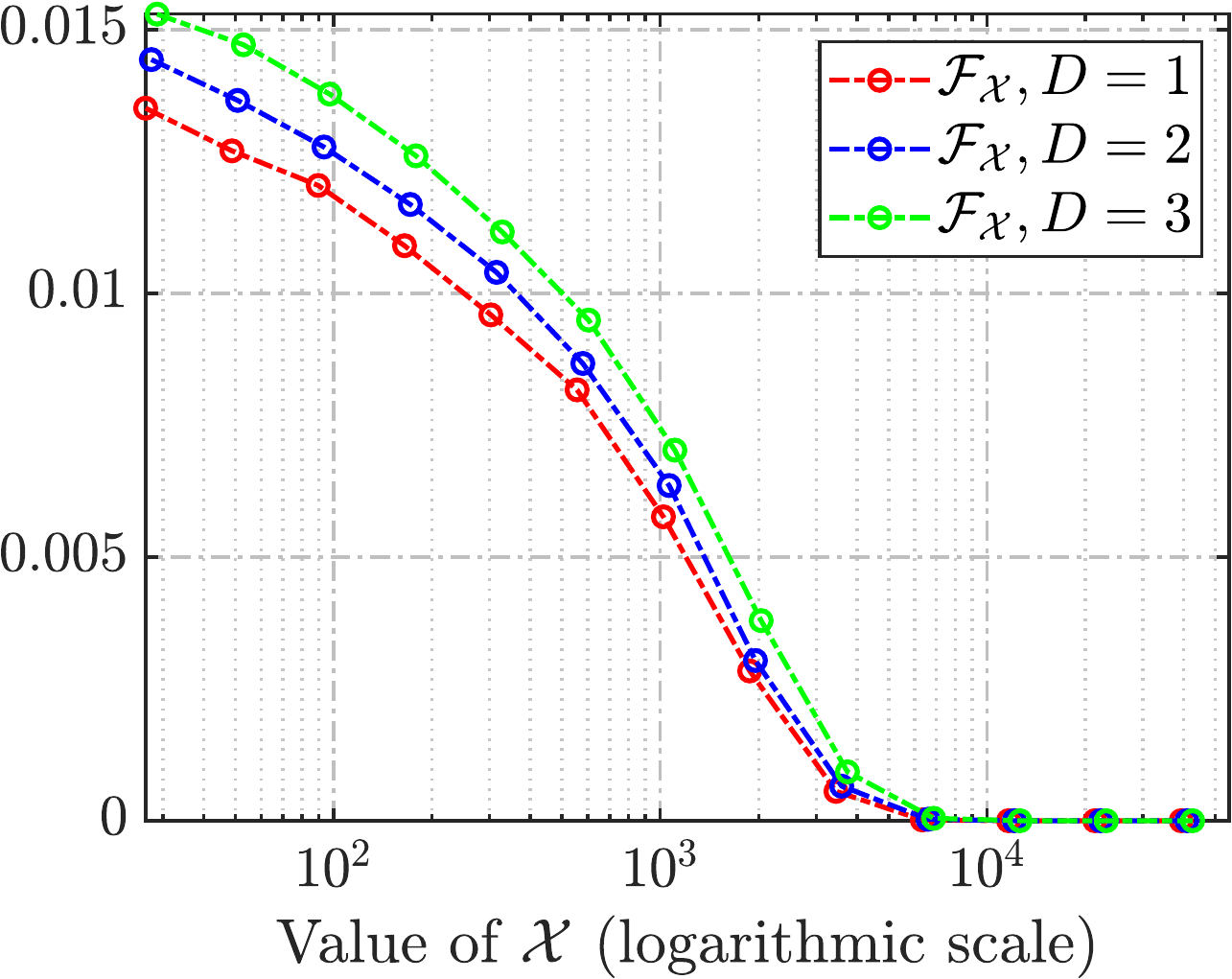}}
  \hspace{0.04\linewidth}%
  \subfloat[\label{f:1d}]{%
		\includegraphics[width=0.47\linewidth]{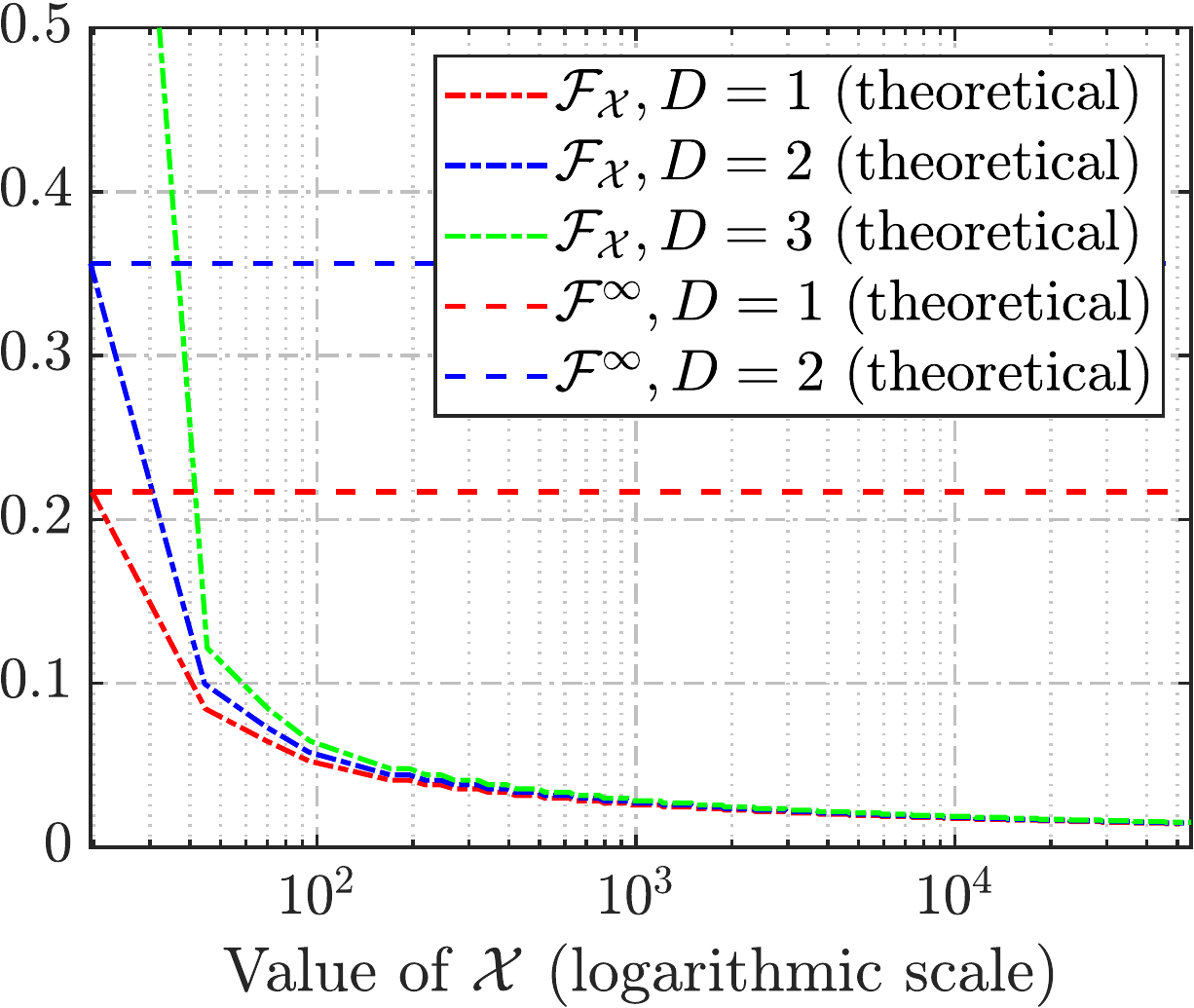}}
  \centering
	\caption{Simulation results various values of $D$. (a) Evolution of
    the second moment of the plant state. (b) Empirical transmission
    fraction and theoretical upper bound on $\tf{\infty}$. (c)
    Empirical \tfsxText. (d) Theoretical upper bounds on $\tf{\infty}$
    and $\tfs{\sx}$.}
	\label{fig1} 
\end{figure}

\section{Conclusion}
\label{sec:conclusion}

In this paper, we have considered a networked control system
consisting of a scalar linear plant with process noise and
non-collocated sensor and controller. Further, the sensor communicates
over a time-varying channel whose state evolves according to an
action-dependent Markov process. The state of the channel determines
the probability with which a packet transmitted by the sensor is
dropped. In this setting, we have designed an event-triggered
transmission policy that guarantees second moment stabilization of the
plant state at a desired rate of convergence to an ultimate bound. We
also derived upper bounds on the transient and the asymptotic
transmission fraction, the fraction of timesteps on which the sensor
transmits. We have verified and illustrated our analysis and
theoretical guarantees through simulations. Future work in this
direction includes incorporation of imperfect measurement of plant and
channel state, application of the proposed action-dependent Markov
channel framework to specific scenarios such as control over a shared
channel and control with energy-harvesting components.

\bibliographystyle{myIEEEtran}
\bibliography{alias,pavan,main}

\appendices

\section{Proofs of Auxiliary Results}
\label{sec:aux-proofs}

\begin{proof}[Proof of Lemma~\ref{lem:matrix-geom-series}]
	First we assume that $\specrad{\mb{K}}<1$. We let $\mb{S}_{\eta}$
  represent the following partial sum
	\begin{equation*}
    \mb{S}_{\eta} \ldef \sum_{w=0}^{\eta-1}\mb{K}^w .
	\end{equation*}
	Observe that $(\id - \mb{K}) \mb{S}_\eta = \id - \mb{K}^\eta$. The
  assumption $\specrad{\mb{K}} < 1$ implies $(\id-\mb{K})$ is
  invertible and hence the following holds
	\begin{align*}
    \mb{S}_{\eta} = (\id - \mb{K})^{-1} (\id - \mb{K}^\eta) .
	\end{align*}
	Again since $\specrad{\mb{K}} < 1$, $\mb{K}^\eta$ converges to zero
  as $\eta \rightarrow \infty$ and thus $\mb{S}_{\eta}$ converges to
  $(\id - \mb{K})^{-1}$ asymptotically.
  
	Now we assume that $\specrad{\mb{K}} \geq 1$. Let $\lambda$ be an
  eigenvalue of $\mb{K}$ with $| \lambda | \geq 1$. Consider a
  corresponding eigenvector $\mb{v}$ (with possibly complex entries).
	\begin{align*}
	\left( \sum_{w=0}^\infty \mb{K}^w \right) \mb{v} = \left(
	\sum_{w=0}^\infty \lambda^w \right) \mb{v} ,
	\end{align*}
  Under the assumption $\specrad{\mb{K}}\geq 1$, the right hand side
  of the above equation diverges and therefore the matrix-geometric
  series $\left(\sum_{w=0}^\infty\mb{K}^w\right)$ also diverges.
\end{proof}

\begin{proof}[Proof of Lemma~\ref{lem:gdk-convergence}]
  From~\eqref{eq:gdk-expand}-\eqref{eq:omega_prob_def}
  and~\eqref{eq:expect-h}-\eqref{eq:gD-def}, we see that an expansion
  of $\gdk$ involves terms such as $g_D(b,\mb{p})$ with $b$ equal to
  $\abar^2, a \abar, a^2$ and $c^2$.  Using~\eqref{eq:param-ordering}
  and noting that $\specrad{b_1\ponee}>\specrad{b_2\ponee}$ when
  $|b_1|>|b_2|$, we can state that
  $\specrad{a^2\ponee}>\specrad{\ti{b}\ponee}$ for $\ti{b}$
  assuming values $\abar^2$, $a\abar$, and $c^2$.  Thus by
  Lemma~\ref{lem:matrix-geom-series},
  $\specrad{a^2\ponee} = a^2\specrad{\ponee} < 1$ is a necessary and
  sufficient condition for convergence of $\gdk$.
\end{proof}

\begin{proof}[Proof of Lemma~\ref{lem:jdsjclosed}]
  Most terms in the closed form of $\gdk$ follow directly
  from~\eqref{eq:gdk-expand}, the series expansion of $\gdk$, the
  closed form of $\Omega_{D}(w,\mb{p})$
  in~\eqref{eq:omega-prob-closed-form}, the expansion of the
  expectation term~\eqref{eq:expect-h}, the
  definition~\eqref{eq:gD-def} and the closed
  form~\eqref{eq:gD-closed-form} of $g_D(b,\mb{p})$. We only need to
  find the closed form of
	\begin{equation*}
    \sum_{w=D}^\infty \max\{ c^{2w}c^{2(k-\Rk)}x^2_{\Rk} , B \}\;
    \Om{D}(b, \pk) .
	\end{equation*}
  We split this summation into two parts based on if $c^{2w} N_k$ is
  larger or smaller than $B$. Observe that $\mu$, defined
  in~\eqref{eq:qkD}, is the smallest integer $w \geq D$ such that
  $B \geq c^{2w}N_k$. Then,
	\begin{align*}
    &\sum_{w=D}^\infty \max\{ c^{2w}c^{2(k-\Rk)}x^2_{\Rk},B \}\;
      \Om{D}(w, \pk) \\ 
    &\; =  g_D(c^2,\pk) \, N_k + \sum_{w=\mu}^\infty
      (B - c^{2w} N_k)\Om{D}(w,\pk)\\ 
    &\; \stackrel{[r1]}{=} Bf_D(1,\pk) + N_k \left[
      g_D(c^2,\pk)-f_D(c^2,\pk)\right] ,
	\end{align*}
  where we obtain [r1] by observing that
  \begin{align*}
    &\sum_{w = \mu}^\infty b^w \Om{D}(w, \mb{p}) = \sum_{w =
      \mu}^\infty b^w \dt (\ponee)^{(w-D)} \pzero^{(D)}
      \mb{p} \\
    &= b^{\mu} \dt (\ponee)^{(\mu-D)} \sum_{w = 0}^\infty 
      (b \ponee)^{w} \pzero^{(D)} \mb{p} = f_D( b , \mb{p} ) ,
  \end{align*}
  assuming $\specrad{b \ponee} < 1$ (see
  Lemma~\ref{lem:matrix-geom-series}). With this we obtain the
  complete closed form expression of the look-ahead function $\gdk$
  provided in Lemma~\ref{lem:jdsjclosed}.
\end{proof}

\begin{proof}[Proof of
  Proposition~\ref{prop:open_loop_performance_function}]
  The proof relies on demonstrating that $H(w,y)>0$ for all
  $w \in \intpos$ and for all $y\in(B,B_0)$.  This implies that if
  $x_{\Sj}^2 \in (B,B_0)$, then the system would violate the online
  objective on the \emph{very} next timestep. From~\eqref{eq:H-expr},
  note that for a fixed $y$, the function $H(w,y)$ can be written as
	\begin{align*}
	H(w,y) = 
	\begin{cases}
	l_1(w,y), \text{ if } w \leq w_{**}(y)\\
	l_2(w,y), \text{ if } w > w_{**}(y),
	\end{cases}
	\end{align*}
  with
  \begin{align*}
    l_1(w,y) &\ldef \abar^{2w} y + \bM (a^{2w}-1) - c^{2w}y \\
    l_2(w,y) &\ldef \abar^{2w} y + \bM (a^{2w}-1) - B ,
  \end{align*}
  where $\displaystyle w_{**}(y) \ldef \frac{\log (y/B)}{\log(1/c^2)}$
  is such that. Now, it suffices to prove the following two claims.

  \emph{Claim (a):} $l_1(w,y)>0$ for all $w\in\intpos$ for
  $y\in(B,B_0)$.
  
  \emph{Claim (b):} $l_2(w,y) > 0$ for all $w\in\intpos$ for
  $y\in(B,B_0)$.

  First, note that $l_1(0,y) = 0$ for all values of $y$. Next,
  evaluating the partial of $l_1(w,y)$ with respect to $w$ at $w=0$
  and for $y\in(B,B_0)$, we obtain
	\begin{align*}
    \frac{ \partial l_1(0,y) }{ \partial w }%
    &= \log(\abar^2/c^2)y + \bM\log(a^2)\\
    & \stackrel{[r1]}{>} \log(\abar^2/c^2)B_0 + \bM \log(a^2)
      \stackrel{[r2]}{=} 0 .
	\end{align*}
  Note that we have used the fact that $\abar^2 < c^2$ to obtain [r1],
  and used the definition of $B_0$ in [r2]. Since $l_1(w,y)$ is a
  quasiconvex function of $w$ (Lemma IV.8,~\cite{PT-MF-JC:2018-tac}),
  it is increasing for all $w > 0$, which proves claim~(a).
	
  Now, we prove claim~(b). We first derive a function $g(w)$ that is a
  lower bound on $l_2(w,y)$ for $w \geq 0$ and $y \in (B, B_0)$.
	\begin{align*}
    l_2(w,y) &= \abar^{2w}y-\bM(a^{2w}-1) - B\\
             &> \abar^{2w}y - y + \bM(a^{2w}-1)\\
             & \stackrel{[r3]}{>} \frac{B_0}{c^{2w}} (\abar^{2w}-1) +
               \bM(a^{2w}-1) \rdef g(w) ,
	\end{align*}
  where in [r3], we have used the fact that
  $\abar^2 < 1$, $c^2 < 1$ and $w \geq 0$.
  Note that $g(w)$ is
  strictly convex in $w$ because
	\begin{align*}
    \frac{ \partial^2 g(w) }{ \partial w^2 }%
    &= B_0\frac{\abar^{2w}}{c^{2w}} \log^2
      (\abar^2/c^2) + \bM a^{2w}\log^2(a^2) > 0.
	\end{align*}
  The derivative of $g(w)$ evaluated at $w=0$ is
  \begin{equation*}
    \frac{ \partial g(0) }{ \partial w } = B_0 \log(\abar^2/c^2) + \bM
    \log(a^2) = 0 , 
  \end{equation*}
  where we have used the definition of $B_0$. Since $g(0) = 0$, $g(w)$
  has slope 0 at $w=0$ and $g$ is strictly convex in $w$, we conclude
  that $l_2(w,y) > g(w) > 0$ for all $w \in \intpos$, which proves
  claim~(b) and thus concluding the proof.
\end{proof}

\begin{proof}[Proof of Lemma~\ref{lem:jdsjclosed2}]
  Recall the infinite series expansion of $\jdsj$
  in~\eqref{eq:jdsj_infinite_sum_form}. To evaluate it, we substitute
  $H(w,x^2_{\Sj})$ with its closed form from~\eqref{eq:H-expr} and
  that of $\ti{\Omega}_D(w,\gamma_{\Sj})$
  from~\eqref{eq:OmegaTildeClosedForm}. Correspondingly, we get an
  expression that is the sum of multiple infinite series, as in the
  derivation of $\gdk$ in Lemma~\ref{lem:jdsjclosed}. To evaluate said
  terms, we define the summation functions $\ti{f}_\theta(b,\gamma)$
  and $\ti{g}_\theta(b,\gamma)$ given in the statement of the lemma
  and which are analogous to $f_\theta(b,\mb{p})$ and
  $g_\theta(b,\mb{p})$, respectively and used for obtaining the
  expression for $\gdk$.
  Proceeding exactly like in Lemma~\ref{lem:jdsjclosed}, we obtain the
  expression for $\jdsj$.
\end{proof}

\begin{proof}[Proof of Proposition~\ref{prop:PEF_upper_bound}]
	We partition the possible values of $x^2_{\Sj}$ into two sets,
	\begin{equation*}
    \llow \ldef [0,Bc^{-2\theta}), \;\; \lhi \ldef [Bc^{-2\theta},\infty),
	\end{equation*}
	and demonstrate that $\J{\Sj}{\theta}<\Mcal{R}_j (\theta)$ in each
  case. The proof is centered around the following two claims, which
  establish bounds on some important terms of the closed form of
  $\J{\Sj}{\theta}$ from Lemma~\ref{lem:jdsjclosed}.
	
	\emph{Claim (a):} If $x^2_{\Sj} \in \llow$, then
	\begin{equation*}
    B\ti{f}_\theta (1,\gsj) \geq \frac{B}{c^{2\theta}} \ti{g}_\theta
    (c^2,\gsj).
	\end{equation*}
	
	\emph{Claim (b):} If $x^2_{\Sj} \in\lhi$, then
	\begin{equation*}
    B\ti{f}_\theta (1,\gsj) \geq x^2_{\Sj} \ti{f}_\theta (c^2,\gsj).
	\end{equation*}
	
	For proving Claim (a), we first recall the term $\nu$ in the closed
  form of $\ti{f}_\theta(b,\gamma)$ from Lemma~\ref{lem:jdsjclosed}
  and note that $\nu=\theta$ when $x^2_{\Sj}< Bc^{-2\theta}$. Thus,
  $\ti{g}_\theta(b,\gsj) = \ti{f}_\theta(b,\gsj)$ when
  $x^2_{\Sj} \in\llow$. We now observe that
	\begin{align*}
    Bc^{-2\theta} \ti{g}_\theta(c^2,\gsj) %
    &\stackrel{[r1]}{=} \frac{B}{c^{2\theta}} c^{2\theta} \dt
      (\mb{I}-c^2\ponee)^{-1}
      \pzero^{(\theta-1)} \pone \canbasis{\gsj} \\
    &\stackrel{[r2]}{\leq} B\dt (\mb{I}-\ponee)^{-1} \pzero^{(\theta-1)} \pone
      \canbasis{\gsj} \\
    &= B\ti{f}_\theta(1,\gsj) ,
	\end{align*}
	where [r1] uses the definition of $\ti{g}_\theta(b,\gsj)$, [r2]
  follows from the fact that the matrix $(\mb{I}-c^2\ponee)^{-1}$ is
  element-wise smaller than $(\mb{I}-\ponee)^{-1}$ because $c^2<1$ and
  all elements of $\ponee$ are non-negative, along with the fact that
  $(\mb{I} - \mb{K})^{-1} = \sum_{w=0}^\infty \mb{K}^w$. This
  completes the proof of Claim (a).
	
	To prove Claim (b), we establish an upper bound on $c^{2\nu}$ under
  the assumption that $x^2_{\Sj} \in \lhi$. Note that
  \begin{align*}
    c^{2\nu} = c^{2\max \left\{ \theta, \left\lceil
    \frac{\log(x^2_{\Sj}/B)}{\log(1/c^2)} \right\rceil
    \right\}} \leq c^{2\left\lceil
    \frac{\log(B / x^2_{\Sj})}{\log(c^2)} \right\rceil } \leq
    \frac{B}{x^2_{\Sj}},
  \end{align*}
  where we have again used the fact that $c^2 < 1$. From this bound,
  one can upper bound $x^2_{\Sj} \ti{f}_\theta (c^2,\gsj)$ as
	\begin{align*}
    x^2_{\Sj} \ti{f}_\theta (c^2, &\gsj) \\
    &\leq B \dt (\ponee)^{(\nu-\theta)} (\mb{I}-c^2\ponee)^{-1}
      \pzero^{(\theta-1)} \pone \canbasis{\gsj} \\
    &\leq B\dt (\ponee)^{(\nu-\theta)} (\mb{I}-\ponee)^{-1}
      \pzero^{(\theta-1)} \pone \canbasis{\gsj} \\
    &= B\ti{f}_\theta (1,\gsj) .
  \end{align*}
	This concludes the proof of Claim (b).
	
	Now, we recall the closed form of $\J{\Sj}{\theta}$. If
  $x^2_{\Sj}\in\llow$, we have
  $\ti{f}_\theta(c^2,\gsj)-\ti{g}_\theta(c^2,\gsj) = 0$ and
  $x^2_{\Sj}<Bc^{-2\theta}$, while
  $\ti{g}_\theta(\bar{a}^2,\gsj) \geq 0$. These facts along with
  Claim~(a) imply that $\J{\Sj}{\theta} \leq \Mcal{R}_j(\theta)$ when
  $x^2_{\Sj} \in \llow$. In the case that $x^2_{\Sj} \in \lhi$, we
  rearrange the closed form of $\J{\Sj}{\theta}$ as
	\begin{align*}
    \J{\Sj}{\theta} = \ %
    &[\ti{g}_\theta (\abar^2,\gsj) - \ti{g}(c^2,\gsj)]x^2_{\Sj} + \bM[
      \ti{g}_\theta (a^2,\gsj)\\ 
    -&\ti{g}_\theta(1,\gsj)] -[B\ti{f}_\theta(1,\gsj) - x^2_{\Sj}
       \ti{f}_\theta(c^2,\gsj) ] .
	\end{align*}
	Then using Claim~(b), the fact that
  $\ti{g}_\theta(\abar^2,\gsj) < \ti{g}_\theta(c^2,\gsj)$ (since
  $\abar^2<c^2$), and lastly the fact that
  $x^2_{\Sj}\geq Bc^{-2\theta}$, we conclude that
  $\J{\Sj}{\theta} \leq \Mcal{R}_j(\theta)$ when
  $x^2_{\Sj}\in\lhi$. Thus, $\Mcal{R}_j(\theta)$ uniformly upper
  bounds $\J{\Sj}{\theta}$ for all $x^2_{\Sj}\in[0,\infty)$.
\end{proof}

\begin{proof}[Proof of Corollary~\ref{cor:infTF}]
	The proof is similar to that of Theorem~\ref{th:stateTF} except for
  one key difference. We note that in Theorem~\ref{th:stateTF},
  $\czero$ was obtained as the $\Mcal{B}$-maximizer of
  $\mb{Q}_{\sx}(D+\Mcal{B})$ under the constraint that
  $\mb{Q}_{\sx}(D+\Mcal{B}) < \mb{0}$. This ensured that the
  transmission fraction over the horizon $\intrangeoc{\Sj}{\Sjnext}$
  is upper bounded by $\cone(\czero+\cone)^{-1}$, under the assumption
  that $x^2_{\Sj}\geq \sx$. In case of asymptotic transmission
  fraction, we know that said upper bound on transmission fraction
  over the horizon $\intrangeoc{\Sj}{\Sjnext}$ has to hold for all
  $j\in\intnonneg$, and equivalently for all $x^2_{\Sj}>0$. Thus we
  derive the term $\czero_{\infty}$ by first maximizing
  $\mb{Q}_{\sx}(D+\Mcal{B})$ over all possible values of $\sx$ and
  then choosing the largest value of $\Mcal{B}$ such that
  $\mb{Q}_{\sx}(D+\Mcal{B})<\mb{0}$ and setting $\czero_{\infty}$
  equal to said value.
	
	The former maximization is carried out because
	$\mb{Q}_{\sx} (D+\Mcal{B})\matrixterms{D+\Mcal{B}}{\gsj}$ acts as an
	upper bound on $\J{\Sj}{D+\Mcal{B}}$, which we want to be negative
	so that~\eqref{eq:intermediate2_tfstate} is valid. Thus, we let
	\begin{align*}
	\czero_\infty \ldef%
	&\argmax\limits_{\Mcal{B}\in\intnonneg} \left\{ \bigg\{\max \limits_{
		\sx \in \real, \ \sx \geq 0}  \mb{Q}_{\sx}( D + \Mcal{B}) \bigg\}
	< \mb{0} \right\}\\
	=& \argmax \limits_{ \Mcal{B} \in\intnonneg } \{ \mb{Q} ( D+\Mcal{B}
     ) < \mb{0} \},
	\end{align*}
  which follows from the fact that $c^2 > \abar^2$ and the definitions
  of $\mb{Q}_{\sx}( \theta )$ and $\mb{Q}( \theta )$. The rest of the
  proof follows along similar lines as that of
  Theorem~\ref{th:stateTF}.
\end{proof}

\section{Procedure to Compute a Sufficient Lower %
  Bound $B^*$ on the Ultimate Bound $B$} \label{sec:Bstar}

Here, we provide a procedure to compute the lower bound $B^*$ on $B$,
referred to in Theorem~\ref{th:sign_monotonicity_H}. This procedure is
based on the proof of Lemma IV.13 in~\cite{PT-MF-JC:2018-tac} and we
present it here for completeness. First, we define the following
constants
\begin{align*}
  P_1 &\ldef \log(a^2/\abar^2), \quad &P_2 &\ldef \log(a^2c^2/\abar^2), \\
  P_3 &\ldef \log(1/c^2), \quad &P_4 &\ldef \log \left(
                                       \frac{\log(1/\abar^2)}{\bM\log(a^2)}
                                       \right) .
\end{align*}
Then, consider the following functions of $B$
\begin{align*}
  &U(B) \ldef e^{(P_3P_4/P_2)} B^{(P_1/P_2)}, \quad w_{**}(U(B)) \ldef
    \frac{\log(B)}{P_2} + \frac{P_2}{P_4}, \\
  &Y(B) \ldef \abar^{2w_{**}(U(B))}U(B) + \bM a^{2w_{**}(U(B))}, \\
  &F_{**}(U(B)) \ldef Y(B) - \bM - B .
\end{align*}
The function $F_{**}(U(B))$ is strictly concave in $B$ (Lemma
IV.13,~\cite{PT-MF-JC:2018-tac}. Thus, it has at most two zeroes, one
of which is 
\begin{equation*}
  B_0 = \frac{\bM\log(a^2)}{\log(c^2/\abar^2)} .
\end{equation*}
There is another zero $B_z > B_0$ of $F_{**}(U(B))$ only if
$F_{**}(U(B))$ is increasing at $B = B_0$. Such a $B_z$ could be found
numerically. Then, we let
\begin{equation*}
  B^* \ldef %
  \begin{cases}
    B_0, \quad \text{if } F_{**}(U(B)) \text{ is non-increasing at } B
    = B_0 \\
    B_z, \quad \text{otherwise} .
  \end{cases}
\end{equation*}

\end{document}